\documentclass[english,journal, draftclsnofoot,12pt, onecolumn]{IEEEtran}
\usepackage[T1]{fontenc}
\usepackage[latin9]{inputenc}
\usepackage{xcolor}
\usepackage{array}
\usepackage{verbatim}
\usepackage{multirow}
\usepackage{amsmath}
\usepackage{amsthm}
\usepackage{amssymb}
\usepackage{graphicx}

\makeatletter

\providecommand{\tabularnewline}{\\}

\theoremstyle{plain}
\newtheorem{lem}{\protect\lemmaname}
\theoremstyle{remark}
\newtheorem{rem}{\protect\remarkname}
\theoremstyle{plain}
\newtheorem{thm}{\protect\theoremname}

\usepackage{amsfonts}
\usepackage{cite}
\usepackage{array}
\usepackage{subfigure}
\usepackage{stfloats}
\usepackage{xcolor}
\usepackage{xpatch}

\makeatletter
\def\changeBibColor#1{
  \in@{#1}{error-9, error-10}
  \ifin@\color {blue}\else\normalcolor\fi
}
\xpatchcmd\@bibitem
  {\item}
  {\changeBibColor{#1}\item}
  {}{\fail}

\xpatchcmd\@lbibitem
  {\item}
  {\changeBibColor{#2}\item}
  {}{\fail}

\makeatother

\usepackage{babel}
\usepackage{caption}
\usepackage{lineno}
\usepackage{graphicx}

\usepackage[caption=false,font=normalsize,labelfont=sf,textfont=sf]{subfig}
\usepackage{caption}
\usepackage{algorithm}
\usepackage{algpseudocode}
\usepackage{algorithmicx}
\usepackage{textcomp}
\usepackage{amsmath}
\usepackage{amsthm}
\usepackage{amssymb}
\usepackage{xcolor}
\usepackage{bm}
\usepackage{color}
\usepackage{multicol}
\makeatother

\usepackage{babel}
\providecommand{\lemmaname}{Lemma}
\providecommand{\remarkname}{Remark}
\providecommand{\theoremname}{Theorem}

\begin{document}
\title{Random Aggregate Beamforming for\\ Over-the-Air Federated Learning in\\
Large-Scale Networks}

\author{Chunmei Xu,~\IEEEmembership{Student~Member,~IEEE},
        Shengheng~Liu,~\IEEEmembership{Member,~IEEE},
        Yongming~Huang,~\IEEEmembership{Senior Member,~IEEE},
        Bj\"orn~Ottersten,~\IEEEmembership{Fellow,~IEEE},
        and~Dusit~Niyato,~\IEEEmembership{Fellow,~IEEE}

\thanks{Manuscript received XXX XX, 2021; revised XXX XX, XXXX; accepted XXX XX, XXXX. Date of publication XXX XX, XXXX; date of current version XXX XX, XXXX.  This work was supported in part by the National Natural Science Foundation of China under Grant Nos. 62001103 and 61720106003, the National Key R\&D Program of China under Grant Nos. 2018YFB1800801 and 2020YFB1806608. (\emph{Corresponding authors: S.~Liu and Y.~Huang}.)}

\thanks{C.~Xu, S.~Liu and Y.~Huang are with the School of Information Science and Engineering, Southeast University, Nanjing 210096, China, and also with the Purple Mountain Laboratories, Nanjing 211111, China (e-mail: \{xuchunmei; s.liu; huangym\}@seu.edu.cn).}

\thanks{B.~Ottersten is with Interdisciplinary Centre for Security, Reliability and Trust (SnT), University of Luxembourg, L-1359 Luxembourg (e-mail: bjorn.ottersten@uni.lu).}

\thanks{D.~Niyato is with the School of Computer Science and Engineering, Nanyang Technological University, Singapore 639798 (e-mail: dniyato@ntu.edu.sg).}

}

\markboth{IEEE TRANSACTIONS ON WIRELESS COMMUNICATIONS,~Vol.~XX, No.~X, XXX~2022}%
{Xu \MakeLowercase{\textit{et al.}}: Random Aggregate Beamforming for Over-the-Air Federated Learning in Large-Scale Networks}

\maketitle
\begin{abstract}
At present, there is a trend to deploy ubiquitous artificial intelligence (AI) applications at the edge of the network. As a promising framework that enables secure edge intelligence, federated learning (FL) has received widespread attention, and over-the-air computing (AirComp) has been integrated to further improve the communication efficiency. In this paper, we consider a joint device selection and aggregate beamforming design with the objectives of minimizing the
aggregate error and maximizing the number of selected devices. This yields a combinatorial problem, which is difficult to solve especially in large-scale networks. To tackle the problems in a cost-effective manner, we propose a random aggregate beamforming-based scheme, which generates the aggregator beamforming vector via random sampling rather than optimization. The implementation of the proposed scheme does not require the channel estimation. We additionally use asymptotic analysis to study the obtained aggregate error and the number of the selected devices when the number of devices becomes large. Furthermore, a refined method that runs with multiple randomizations is also proposed for performance improvement. Extensive simulation results are presented to demonstrate the effectiveness of the proposed random aggregate beamforming-based scheme as well as the refined method.
\end{abstract}

\begin{IEEEkeywords}
Federated learning, over-the-air computation (AirComp), device selection, aggregate beamforming, large-scale distributed systems.
\end{IEEEkeywords}

\section{Introduction}

The ubiquitous mobile gadgets and Internet of Things (IoT) devices,
together with the recent revival and breakthrough of artificial intelligence
(AI), have inspired people to envision intelligence at the edge of the wireless networks \cite{edge, Platform}. Many AI tasks are computationally intensive, which are conventionally trained at a powerful server center with a large amount of data collected and stored. However, the centralized-training paradigm is generally connected to high latency, because of the big volume of data to be transmitted. In addition, collecting sensitive data such as geographic locations \cite{location_privacy} and personal health records \cite{eHealth} can lead to serious privacy violations. Fortunately, mobile edge computing (MEC) brings storage and computation resources closer to the user and allows circulation of data to take place locally \cite{MEC}. Moreover, the latest mobile devices
have already been equipped with high-performance graphical-processing units
(GPUs), which can be used for task training. Riding on these trends,
the deployment of AI algorithms at the network edge is feasible.
On the other hand, to address the privacy concerns in distributed learning, the framework of federated learning (FL) is surging in popularity, which collaboratively trains a globally shared model without sharing the sensitive raw data \cite{2016communication}. Despite effective prevention of data leakage, communication overhead is a principal bottleneck of the FL schemes, since frequent model transmission via wireless links are required \cite{edge,chen2019FLwireless,yang2019eeFL}.

It is widely foreseen that the future wireless network service such as IoT will leverage dense edge devices providing pervasive sensing and actuation ability. For a large-scale system with a massive number of devices, normally only a small fraction of devices are selected for local model updating. The reasons are twofold respectively from machine learning and communication perspectives. First, it is observed that the learning performance may degrade when excessive devices participate in the global model aggregation \cite{2016communication}, as too large sized batches lead to higher validation error \cite{Large_Minibatch_SGD}. Second, aggregation is constrained by the latency requirement and available wireless resources, which restricts the number of
devices involved. Hence, elaborate design of device selection and transmission mechanism is of paramount importance for edge intelligence in future large-scale wireless networks.

For device selection, the original protocol was designed to randomly
select a fraction of devices for local model updating of FL \cite{2016communication}.
However, this approach is inefficient in practical settings, since heterogeneous devices
have different reliability values, data sizes, computation/storage
capacities, and experience distinct wireless channels. For instance,
unreliable devices may perform malicious updates, which degrades the learning performance, as FL is susceptible to adversarial attacks \cite{kang2020reliable,kang2019incentive}.
Besides, the \emph{stragglers} with low storage
and computation capacities or appalling channel conditions generally require a longer
time to upload their local models. In a synchronized FL
system, the aggregator has to wait or ignore the stragglers, which
in return impairs the learning efficiency
\cite{bonawitz2019federated,li2020federated}. In this context, three simple scheduling policies were implemented in \cite{scheduling} and the corresponding FL convergence rates were analyzed by taking into account the wireless channel conditions \cite{scheduling}.
Further, Nishio et al. \cite{client_seletion_heter} consider maximizing the number of the selected devices under the training time budget constraints, which requires the collection of the available resource information of the devices. However, device selection schemes in the above studies are not assisted by proper transmission mechanism that can further improves the communication efficiency.

Improving the communication efficiency between the selected devices and the aggregator helps reduce the communication and training time budgets. Previous related works have considered reducing the time budgets by cutting down the size of the transmitted data via data quantization \cite{quantization} or model sparsification \cite{compression}. Stich et al. also proposed to aggregate the local models less frequently, where the global model is attained after the local models are updated for multiple steps rather than after each iteration \cite{local_sgd}. Additionally, the communication efficiency can be improved by lifting the transmission rate via advanced communication technologies \cite{Broadband}. Conventionally, in data aggregation, the receivers need to decode each individual transmitted data before further computations. This pattern is commonly referred to as the separated-communication-and-computation principle. Based on this conventional principle, the problem of joint device selection and bandwidth allocation is addressed \cite{fast_convergence}, where the convergence rate is maximized under the time budget constraint determined by the achievable computation latency. As mitigating the generated interferences consumes more wireless resources, this scheme is inefficient especially given the limited resources. As an alternative, the emerging
over-the-air computation (AirComp) principle \cite{zhu2020overtheair, AirComp} integrates computation and communication by leveraging the waveform
superposition property \cite{Gunduz,kaiyang,Broadband,xu2021learning}. The over-the-air analog aggregation can significantly improve the communication efficiency compared to the digital manner \cite{Broadband, Gunduz}

Nevertheless, AirComp introduces signal distortion due to the fading channels and noises, which can be measured by mean square error (MSE). Such error deviates the aggregate data from the desired one, which harms FL learning performance as well as the communication efficiency \cite{Broadband}. Specifically, the existence of error implies that the received model parameters at the aggregator are perturbed. A large perturbation can lead to divergence of the training loss and degrade the classification/regression performance \cite{effect_error}. Although error control in the upper layer can be applied, it inevitably lowers the communication efficiency. Thus, the aggregate error must be mitigated. To this end, power allocation for over-the-air FL is considered \cite{Power_Air1} by taking gradient statistics into account. Utilizing the multiple antennas equipped at the aggregator, the aggregate beamforming vector and the scaling factor were designed by incorporating dynamic learning rates \cite{xu2021learning}. Zhu et al. \cite{Broadband} considered joint device selection and transmit power design and proposed the method based on the derived communication-and-learning tradeoffs. Essentially, the optimization of transmit power, aggregate beamforming and other resource allocation targeting at MSE reduction is to align the signals from the selected devices. On the other hand, a moderate perturbation on the model parameters is beneficial. It avoids over-fitting and improves the robustness of a neural network against adversarial attacks, which can be regarded as a regulation technique \cite{effect_noise1, effect_noise2, effect_noise3, effect_noise4}. In this regard, selecting suitably more devices for local model updates and global model aggregation can improve the learning efficiency, since large-sized equivalent minibatch speeds up the training process while guaranteeing the achieved learning performance \cite{Large_Minibatch_SGD,fast_convergence}. Aiming at maximizing the number of the selected devices, Yang et al. \cite{kaiyang} considered the problem of joint device selection and aggregate beamforming vector optimization under the aggregate error constraint by using the difference of convex (DC).

Design of joint device selection and transmission scheme for massive number of devices deploying AI applications is a meaningful topic. As the underlying combinatorial problems in the large-scale systems are extremely difficult to solve, it is challenging and, to the best of our knowledge, remains unexplored yet. The size of the selected device subset grows exponentially with the number of devices, which makes exhaustive search prohibitive. Despite the simplicity of the random device selection
protocol, the resulting aggregate error can be unacceptable. As a compromise,
conventional optimization methods can achieve good performance, but
the required computational complexity may still be high. Besides, the acquisition of channel information in a large-scale system induces extremely high estimation overhead. Therefore, a new cost-effective method that can still achieve adequate performance is required. In this paper, we investigate the joint device selection and aggregate beamforming design in a large-scale system (e.g., IoT), where the FL framework and the AirComp technique are adopted. The heterogeneity of the devices focuses on the wireless channel states. Two relevant objectives, i.e., MSE minimization with a fixed number of devices, and the number of selected devices maximization under the MSE constrain, are considered for the systems with different requirements/goals, which make this paper comprehensive. The contributions of this paper are elaborated as follows:

\begin{itemize}

\item We propose a random aggregate beamforming-based scheme to solve the two combinatorial problems in a large-scale system. The core idea is to uniformly sample a vector from the complex unit sphere as the aggregate beamforming vector first, and select the devices afterwards. The implementation complexity of the proposed methods is low, which does not require the channel state information.

\item Through asymptotic analysis, we prove that the minimum aggregate error can be approached using the proposed methods when the number of devices becomes large. We also derive the number interval as well as the average value for the problem  of maximizing the number of the selected devices.

\item To improve the obtained performance of both the problems in practical scenarios where the number of devices is less than infinity, we propose the refined methods that samples multiple vectors from the complex unit sphere. Although the required number of the vectors that achieves acceptable performance cannot be explicitly given, it provides an insight to obtain solutions with performance improvement.

\end{itemize}

The remainder of the paper is organized as follows. Section II gives
the system model, where FL and AirComp technique are detailed. Section
III formulates the investigated problems including the minimization
of the aggregate error and the maximization of the number of the selected
devices. In Section IV, we propose the aggregate beamforming based
methods for both considered problems and discuss their theoretical
performance, and the refined methods are given in Section V for performance
improvement. Simulation results are presented in Section VI and the
paper is concluded in Section VII.

\section{System Model }

We consider a large-scale wireless network with edge intelligence,
which consists of $K$ single-antenna devices and an aggregator equipped
with $N \ll K$ antennas. Each device $k$ owns a local dataset $\mathcal{D}_{k}$,
which collectively constitutes the global dataset $\mathcal{D}=\cup_{k\in\mathcal{K}}\mathcal{D}_{k}$.
Typically, the objective function of a learning task is to find the
model $\mathbf{w}^{o}$ that minimizes the loss function $P\left(\mathbf{w};\mathcal{D}\right)$,
where $\mathbf{w}\in\mathbb{R}^{D}$ is the model weight parameters
with the dimension of $D$. Although the centralized gradient decent
method can be applied to update the model $\mathbf{w}$ in order
to obtain $\mathbf{w}^{o}$, it becomes impractical as local datasets $\mathcal{D}_{k}$ at each device cannot be accessed by others
due to the privacy and latency concerns. To handle these concerns,
the FL can be adopted.

The iterative learning process of FL includes: 1) A subset of devices are selected; 2) The selected devices update the model with local data, and upload the updated model back to the aggregator; 3) The aggregator receives these local models and aggregates them to obtain the global model.
As shown
in Fig. \ref{fig:System-model}, each device $k$ updates its model
based on the locally stored dataset $\mathcal{D}_{k}$, and transmits
the local model back to the aggregator. Then, the global model is
updated by averaging the local models from the participating devices,
which is referred to as the model aggregation phase. Following the updating
rules of FL, the local models $\mathbf{w}_{k},k\in\mathcal{K}$
and the global model $\mathbf{w}$ are respectively updated by
\begin{equation}
\mathbf{w}_{k}^{i+1}=\mathbf{w}^{i}-\mu\mathbf{g}_{k}\left(\mathbf{w}^{i}\right),
\end{equation}
and
\begin{equation}\label{eq:update}
\mathbf{w}^{i+1}=\frac{1}{K}\sum_{k=1}^{K}\mathbf{w}_{k}^{i+1},
\end{equation}
where superposition $i$ is the iteration index, $\mu$ is the learning
rate or step size. $\mathbf{g}_{k}(\mathbf{w}^{i})=\nabla_{\mathbf{w}^{i}}P_{k}\left(\mathbf{w}^{i}\right)$
is the gradients of the loss on the local dataset $\mathcal{D}_{k}$
with respect to (w.r.t.) $\mathbf{w}^{i}$, where $P_{k}\left(\mathbf{w}^{i}\right)=\frac{1}{\left|\mathcal{D}_{k}\right|}\sum_{n=1}^{\left|\mathcal{D}_{k}\right|}Q\left(\mathbf{w}^{i};\mathcal{D}_{k}^{n}\right)$
is the loss function with $Q\left(\mathbf{w}^{i};\mathcal{D}_{k}^{n}\right)$
the loss value on the $n$-th data sample $\mathcal{D}_{k}^{n}$.

\begin{figure}
\centering
\includegraphics[width=0.6\columnwidth]{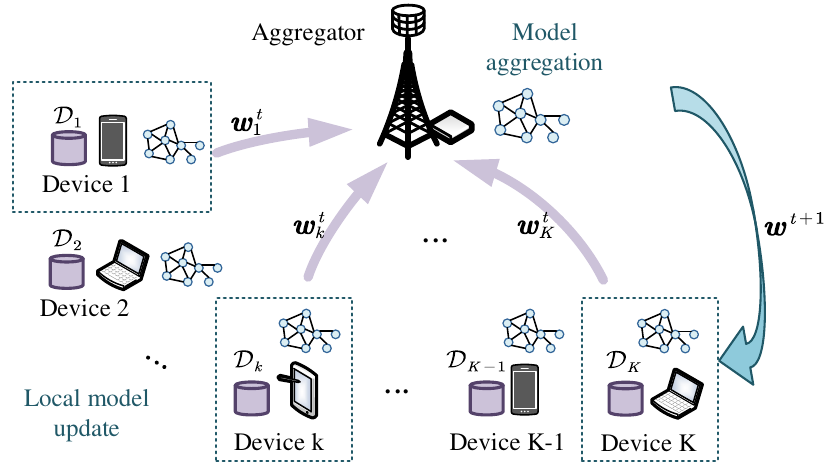}
\captionsetup{font={small}}\caption{System hierarchy model under investigation.} 
\label{fig:System-model}
\end{figure}

Apart from the observation from \cite{2016communication}, the limited
wireless resources also deter the participation of all devices in
updating the global model. Thus, only a subset of devices are selected
for local models updating and the global model aggregation. In the
considered model, AirComp technique is applied for the model aggregation
in order to improve the communication efficiency. The transmit symbol
vector of device $k$ at the $i$-th iteration is customized as $\mathbf{s}_{k}^{i}\triangleq\mathbf{w}_{k}^{i}$,
which is assumed to be normalized with unit variance, i.e., $\mathbb{E}\left[\mathbf{s}_{k}^{i}\left(\mathbf{s}_{k}^{i}\right)^{\mathrm{H}}\right]=\mathbf{I}$.
For notational convenience, the $d$-th element of $\mathbf{s}_{k}^{i}$,
$\mathbf{w}^{i}$ and $\mathbf{g}_{k}$, i.e., $\mathbf{s}_{k}^{i}[d]$,
$\mathbf{w}^{i}\left[d\right]$ and $\mathbf{g}_{k}[d]$,
are denoted as $s_{k}$, $w$ and $g_{k}$. By denoting the set of selected
device as $\mathcal{S}$, the desired signal based on (\ref{eq:update})
is written as
\begin{equation}
y_{\mathrm{des}}=\sum_{k\in\mathcal{S}}w_{k}=\sum_{k\in\mathcal{S}}s_{k}.
\end{equation}

In implementing AirComp, the symbols $s_{k}$ are modulated in an
analog manner and precoded by the transmit coefficients $b_{k}$. The amplitude of $b_k$ means the transmit power of device $k$ and its phase is used to help align the obtained  signals at the aggregator. Then,  these signals are superimposed over the air, and the received signals at the aggregator are combined by the aggregate beamforming
vector $\mathbf{m}$. Finally, the resulted signal is further
amplified by a factor $\eta$. Due to the fading and noisy wireless
channels, the actual received signal at the aggregator is given by
\begin{equation}\label{eq:aggregat_sig}
y=\sqrt{\eta}\left(\sum_{k\in\mathcal{S}}\mathbf{m}^{\mathrm{H}}\mathbf{h}_{k}b_{k}s_{k}+\mathbf{m}^{\mathrm{H}}\mathbf{n}\right),
\end{equation}
where $\mathbf{h}_{k}$ is the channel vector from device $k$
to the aggregator, which is assumed to be Complex Gaussian distributed
with unit power, i.e., $\mathbf{h}_{k}\sim\mathcal{CN}\left(\mathbf{0},\mathbf{I}\right)$.
This suggests that the channels are independently identically distributed (i.i.d.). Vector $\mathbf{n}$
is the additive white Gaussian noise such that $\mathcal{CN}\left(0,\sigma^{2}\right)$.
Thus, the resulted aggregate error via AirComp is expressed as
\begin{equation}\label{eq:error}
e\triangleq y_{\mathrm{des}}-y=\sum_{k\in\mathcal{S}}\left(1-\sqrt{\eta}\mathbf{m}^{\mathrm{H}}\mathbf{h}_{k}b_{k}\right)s_{k}-\sqrt{\eta}\mathbf{m}^{\mathrm{H}}\mathbf{n},
\end{equation} 
Letting $a_{k}^{\mathrm{}}=\left(1-\sqrt{\eta}\mathbf{m}^{\mathrm{H}}\mathbf{h}_{k}b_{k}\right)$,
the MSE is then written as
\begin{align}\label{eq:MSE-1}
\mathrm{MSE} & =\mathbb{E}\left(\left\Vert e\right\Vert ^{2}\right)\nonumber
  =\mathbb{E}\left[\left(\sum_{k\in\mathcal{S}}a_{k}s_{k}-\sqrt{\eta}\mathbf{m}^{\mathrm{H}}\mathbf{n}\right)^{\mathrm{H}}\left(\sum_{k\in\mathcal{S}}a_{k}s_{k}-\sqrt{\eta}\mathbf{m}^{\mathrm{H}}\mathbf{n}\right)\right]\nonumber \\
 & =\underset{\mbox{fading-related error}}{\underbrace{\sum_{k\in\mathcal{S}}a_{k}^{\mathrm{H}}a_{k}}}+\underset{\mbox{noise-related error}}{\underbrace{\eta\left\Vert \mathbf{m}\right\Vert ^{2}\sigma^{2}}},
\end{align}
which consists of the fading-related and the noise-related components.
To reduce the error, we can eliminate the error due to the fading while mitigating the error caused by the noise. According to (\ref{eq:MSE-1}), the elimination
of error caused by the fading should guarantee the following condition \cite{kaiyang,xu2021learning}:
\begin{equation}
a_{k}=1-\sqrt{\eta}\mathbf{m}^{\mathrm{H}}\mathbf{h}_{k}b_{k}=0,k\in\mathcal{S}.
\end{equation}
Consequently, the resulted MSE are respectively expressed as
\begin{equation}
\mathrm{MSE}=\mathbb{E}\left(\left\Vert e\right\Vert ^{2}\right)=\eta\left\Vert \mathbf{m}\right\Vert ^{2}\sigma^{2}.\label{eq:MSE}
\end{equation}

\section{Problem Formulation}

When applying the AirComp technique, the signal distortion due to the fading and the noise results in the deviation of the aggregate data from the true one, which is critical for the performance of FL tasks. In the scenarios where the signal to noise ratio (SNR) is low, the resulted aggregate error under the fixed number of the selected devices can be extremely lager. Such large error degrades the learning performance \cite{effect_error,Broadband}, and hence it should be reduced. On the other hand, a considerable SNR would lead to a moderate perturbation which helps improve the robustness of the neural networks \cite{effect_noise1,effect_noise2,effect_noise3}. Thereby,  more devices need to be selected for local model update in order to enhance the learning efficiency \cite{Large_Minibatch_SGD, Broadband}. In this paper, we investigate joint device selection and aggregate beamforming design with two objectives of MSE minimization and the number of selected  devices maximization. The problem formulations are presented in this section.
As we have assumed before that the devices are homogeneous except
for the fading channels, the device selection policy considered here
is dependent on the communication factors. The consideration of learning
factors such as data size, computation and storage capacities, will be left for our future works.

\subsection{MSE Minimization}

To maintain the training and inference performance of the AI tasks,
the large aggregate error measured by MSE should be minimized. Note that
the number of the devices selected for the model is fixed, i.e., $\left|\mathcal{S}\right|=S$,
where $\left|\mathcal{S}\right|$ is the cardinality of subset $\mathcal{S}$. The
optimization variables include the scaling factor $\eta$, transmitting
coefficient \textbf{$b_{k}$}, device subset $\mathcal{S}$, and aggregate
beamforming vector $\mathbf{m}$. Since the above variables are
independent from noise vector $\mathbf{n}$, minimizing the MSE in (\ref{eq:MSE})
shares the identical solution with the objective of $\eta\left\Vert \mathbf{m}\right\Vert ^{2}$.
Mathematically, the MSE minimization problem is formulated as
\begin{subequations}\label{eq:SIMO}
\begin{align}
\min_{\mathbf{m},\eta,b_{k},\mathcal{S\subseteq K}}\quad & \left\Vert \mathbf{m}\right\Vert ^{2}\eta\label{eq:pro1_1-1}\\
\mathrm{s.t.\quad} & \sqrt{\eta}\mathbf{m}^{\mathrm{H}}\mathbf{h}_{k}b_{k}=1,k\in\mathcal{S}\label{eq:pro1-2-1}\\
 & \left|\mathcal{S}\right|=S\label{eq:pro1_3-1}\\
 & \left\Vert b_{k}\right\Vert ^{2}\leq P,k\in\mathcal{S}\label{eq:pro1_4-1}
\end{align}
\end{subequations}
where constraints (\ref{eq:pro1-2-1}), (\ref{eq:pro1_3-1}) and (\ref{eq:pro1_4-1})
are respectively the condition of eliminating the fading-related error,
the fixed number of the selected devices, and the transmit power constraint
with $P$ the maximum power.

The optimal transmit coefficient $b_{k}$ following \cite{8364613}
can be designed as
\begin{equation}
b_{k}=\frac{\mathbf{h}_{k}^{\mathrm{H}}\mathbf{m}}{\sqrt{\eta}\left\Vert \mathbf{m}^{\mathrm{H}}\mathbf{h}_{k}\right\Vert ^{2}}.\label{eq:transmit}
\end{equation}
Substituting back to constraint (\ref{eq:pro1_4-1}) yields $\eta\geq\frac{1}{P\left\Vert \mathbf{m}^{\mathrm{H}}\mathbf{h}_{k}\right\Vert ^{2}},k\in\mathcal{S}$
for each device $k$. Thereby, $\eta$ is derived as
\begin{equation}
\eta=\underset{k}{\max}\frac{1}{P\left\Vert \mathbf{m}^{\mathrm{H}}\mathbf{h}_{k}\right\Vert ^{2}},k\in\mathcal{S},\label{eq:yeta}
\end{equation}
which transforms problem (\ref{eq:SIMO}) into
\begin{align}
\min_{\mathbf{m},\mathcal{S\subseteq K}}\quad & \underset{k\in\mathcal{S}}{\max}\frac{\left\Vert \mathbf{m}\right\Vert ^{2}}{P\left\Vert \mathbf{m}^{\mathrm{H}}\mathbf{h}_{k}\right\Vert ^{2}}\quad\mathrm{s.t.}\;(\mathrm{\mathrm{\ref{eq:pro1_3-1}}}).\label{eq:SIMO1}
\end{align}
Since both numerator and denominator of the above objective contain
the aggregate beamforming vector $\mathbf{m}$, we can add an
extra constraint such that $\left\Vert \mathbf{m}\right\Vert =1$,
which would not change the objective value. Moreover, since  $\max \min 1/=\min \max$, the above problem is equivalent to
\begin{subequations}\label{eq:SIMO1_1}
\begin{align}
\max_{\mathbf{m},\mathcal{S\subseteq K}}\quad & \underset{k\in\mathcal{S}}{\min}P\left\Vert \mathbf{m}^{\mathrm{H}}\mathbf{h}_{k}\right\Vert ^{2}\label{eq:SIMO2-1-4}\\
\mathrm{s.t.}\mathrm{\quad} & \left\Vert \mathbf{m}\right\Vert =1\label{eq:SIMO2-2}\\
 & (\mathrm{\mathrm{\ref{eq:pro1_3-1}}}),\nonumber
\end{align}
\end{subequations}
which is a mixed combinatorial optimization problem.

\subsection{Maximizing The Number of Selected Devices }

A moderate perturbation caused by the error can be interpreted as a regulation technique, which enhances the robustness of a neural network. More devices involving in local model updating
and global model aggregation
can improve the learning efficiency \cite{2016communication}\cite{Large_Minibatch_SGD}.
Here, we consider the maximum tolerance error denoted as $\overline{\mathrm{MSE}}$.
The objective is to maximize the number of the selected devices under
the MSE constraint. Same as the MSE minimization problem, the transmitting
coefficient \textbf{$b_{k}$} and the scaling factor $\eta$ are designed
respectively as (\ref{eq:transmit}) and (\ref{eq:yeta}). Defining
$\mathrm{MSE}_{k}=\frac{\left\Vert \mathbf{m}\right\Vert ^{2}\sigma^{2}}{P\left\Vert \mathbf{m}^{\mathrm{H}}\mathbf{h}_{k}\right\Vert ^{2}}$
where $\left\Vert \mathbf{m}\right\Vert =1$, the resulted MSE
in (\ref{eq:MSE}) can be expressed as $\mathrm{MSE}=\underset{k\in\mathcal{S}}{\max}  \mathrm{MSE}_{k}$. Mathematically, the problem is formulated as
\begin{subequations}\label{eq:SIMO3}
\begin{align}
\max_{\mathbf{m},\mathcal{S\subseteq K}}\quad & \left|\mathcal{S}\right|\label{eq:SIMO3_1}\\
\mathrm{s.t.}\mathrm{\quad} & \mathrm{MSE}=\underset{k\in\mathcal{S}}{\max}\mathrm{MSE}_{k}\leq\overline{\mathrm{MSE}}\label{eq:SIMO3_2}\\
 & (\ref{eq:SIMO2-2}),\nonumber
\end{align}
\end{subequations}which is also a mixed combinatorial optimization problem. The constraint
(\ref{eq:SIMO3_2}) therein indicates that the aggregate error should
be less than or equal to $\overline{\mathrm{MSE}}$. 
communication performance.

Both considered problems are mixed combinatorial, which are difficult
to solve. To obtain the optimal selected device subset, it is straight-forward
to search the device subset in a brute-force manner. However, the
size of the search space, i.e., $\tbinom{K}{S}$, becomes prohibitively
large when $K$ is large, which makes brute-force searching intractable.
The device subset can be determined by utilizing the sparse property
as in \cite{kaiyang}, but the computational burden can be extremely
large in a network with a massive number of edge devices. Moreover,  optimizing the aggregate beamforming $\mathbf{m}$ for the
selected devices
is still challenging due to the nonconvexity. 
Therefore, new methods with low-complexity are needed to address
these issues  in a large-scale wireless network.

\section{Random Aggregate Beamforming-Based Scheme}

In this section, we propose a random aggregate beamforming-based scheme to solve the previous formulated problems, which is cost-effective. We also give the theoretical analysis in terms of the obtained MSE and
the number of the selected devices. With regard to $\mathbf{m}$, we
arrive at the following lemma.
\begin{lem}
\label{thm:infinity-1}For any given feasible $\mathbf{m}$ and
arbitrary $\theta\in\mathbb{R}$, the objectives of the considered problems
(\ref{eq:SIMO1_1}) and (\ref{eq:SIMO3}) are identical under $\mathbf{m}$
\textup{and $\mathbf{m}e^{j\theta}$.}
\end{lem}
\begin{proof}
For any $\theta\in\mathbb{R}$, we have $\left\Vert \left(\mathbf{m}e^{j\theta}\right)^{\mathrm{H}}\mathbf{h}_{k}\right\Vert ^{2}=\left\Vert \left(\mathbf{m}\right)^{\mathrm{H}}\mathbf{h}_{k}\right\Vert ^{2}$
and $\left\Vert \mathbf{m}e^{j\theta}\right\Vert =\left\Vert \mathbf{m}\right\Vert $,
which indicate that solutions $\mathbf{m}e^{j\theta}$ and $\mathbf{m}$
have the same objective values and guarantee the constraints.
\end{proof}
\begin{rem}
\label{thm:infinity} There are infinite solutions of the aggregate
beamfoming vector $\mathbf{m}$ with the same objectives. Typically,
if $\mathbf{m}^{*}$ is the optimal solution to problems (\ref{eq:SIMO1_1})
and (\ref{eq:SIMO3}), $\mathbf{m}^{*}e^{j\theta},\forall\theta\in\mathbb{R}$
is also optimal, which indicates that there are infinite optimal solutions
of $\mathbf{m}$.
\end{rem}

\subsection{Random Aggregate Beamforming-Based Methods}

We observe that if the aggregate beamforming vector $\mathbf{m}$
is determined at the very first, the optimal devices subset for the
investigated problems can be easily obtained by arranging and sorting
the equivalent channel power $\left\Vert \mathbf{m}^{\mathrm{H}}\mathbf{h}_{k}\right\Vert ^{2}$.
Obtaining optimal $\mathcal{S}$ in this manner only requires the
computational complexity of $\mathcal{O}\left(K\right)$, which is
computationally efficient. Inspired by this, we propose the random
aggregate beamforming-based scheme to the investigated problems.
The core idea is to determine the aggregate beamforming vector $\mathbf{m}$
first, and the selected device subset $\mathcal{S}$ afterwards.
Specifically, sample a vector from the set, i.e.,
\[
\mathcal{M}=\left\{ \mathbf{m}|\text{\ensuremath{\left\Vert {\mathbf{m}}\right\Vert }=1},{\mathbf{m}}\in\mathbb{C}^{N}\right\} ,
\]
which is a complex unit sphere with $N$ dimensions. To uniformly
generate $\mathbf{m}$, we can normalize a random vector $\mathbf{m}_{1}$
from $\mathcal{CN}\left(\mathbf{0},\mathbf{I}\right)$, i.e.,
$\mathbf{m}=\mathbf{m}_{1}/\left\Vert \mathbf{m}_{1}\right\Vert ,\mathbf{m_{1}}\sim\mathcal{CN}\left(\mathbf{0},\mathbf{I}\right)$,
whose computational overhead is negligible. For problem (\ref{eq:SIMO1_1}),
the devices subset is optimized by arranging the equivalent channel
power $\left\Vert \mathbf{m}^{\mathrm{H}}\mathbf{h}_{k}\right\Vert ^{2}$
and selecting the devices with the $S$ largest values. Similarly, for
problem (\ref{eq:SIMO3}), the devices are selected if satisfying
the condition of $\mathrm{MSE}_{k}\leq\overline{\mathrm{MSE}}$. It is worth noting that the proposed methods have significantly low complexity,
since generating vector $\mathbf{m}$ and obtaining the subset
$\mathcal{S}$ therein are easy.
\begin{figure}
\centering
\subfigure[] {\includegraphics[width=0.48\columnwidth]{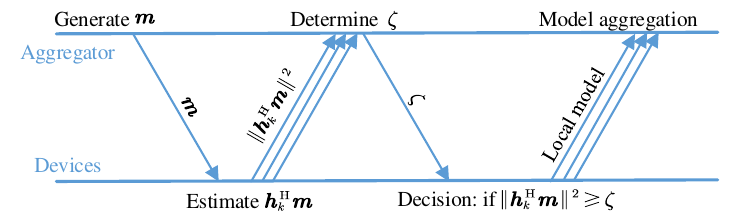}}
\subfigure[] {\includegraphics[width=0.48\columnwidth]{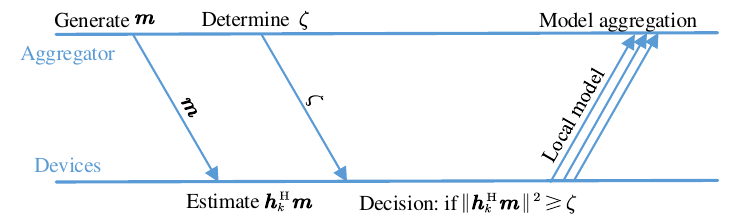}}	
\captionsetup{font={small}} 
\caption{Implementation of the proposed scheme for (a) problem (\ref{eq:SIMO1_1}) and (b) problem (\ref{eq:SIMO3}).}
\label{fig:implement}
\end{figure}

The implementation of the proposed algorithms to minimize MSE and maximize the number of selected devices is shown in Fig. \ref{fig:implement}. According to this figure, the implementation of our random-based algorithms does not require the channel estimation between the aggregator and the edge devices, which can greatly reduce the implementation complexity. For MSE minimization problem (\ref{eq:SIMO1_1}), the aggregator first generates a random aggregate beamforming vector $\mathbf{m}$, which is broadcast to all devices. Then all devices estimate the obtained value $\mathbf{h}^{\mathrm{H}}_k\mathbf{m}$ which is a weighted channel value  rather than $\mathbf{h}^{\mathrm{H}}_k$, and feedback $\Vert\mathbf{h}^{\mathrm{H}}_k\mathbf{m}\Vert^2$ to the aggregator.  After receiving the  feedbacks from all devices, the aggregator determines and broadcasts a threshold $\zeta$, which is  obtained by finding the $S$ largest weighted channel gain values. Finally, the devices are selected for local model updates and global model aggregation if their weighted channel gain values are no less than $\zeta$. For problem (\ref{eq:SIMO3}), the implementation procedure is similar except for the determination of  $\zeta$. The threshold therein is determined directly at the aggregator without the information of weighted channel gains, which is obtained such that $\zeta=\frac{\sigma^2}{P\overline{\mathrm{MSE}}}$.
The implementation of our proposed algorithms does not require any channel information.  This is the major advantage compared with those methods requiring the channel information, since the overhead for channel estimation in a lager-scale network is heavy.

In the following, we focus on the theoretical
analysis in terms of the objectives of MSE and the number of the selected
devices in a large-scale system with massive edge devices.

\subsection{The Distribution of $\left\Vert \mathbf{m}^{\mathrm{H}}\mathbf{h}_{k}\right\Vert $}

We first analyze the property of inner product $\mathbf{m}^{\mathrm{H}}\mathbf{h}_{k}=\sum_{i=1}^{N}\mathbf{m}[i]\mathbf{h}_{k}[i]$
where $\mathbf{m}[i]$ and $\mathbf{h}_{k}[i]$ are the $i$-th
element of vector $\mathbf{m}$ and $\mathbf{h}_{k}$. Since
\textbf{$\mathbf{h}_{k}[i]$ }is Gaussian distributed, i.e., \textbf{$\mathbf{h}_{k}[i]\sim\mathcal{CN}(0,1)$},
we have $\mathbf{m}[i]\mathbf{h}_{k}[i]\sim\mathcal{CN}(0,\left\Vert \mathbf{m}[i]\right\Vert ^{2})$.
Together with the additional constraint of $\left\Vert \mathbf{m}\right\Vert =1$,
the distribution of $\mathbf{m}^{\mathrm{H}}\mathbf{h}_{k}$
is readily obtained as
\begin{equation}
\mathbf{m}^{\mathrm{H}}\mathbf{h}_{k}\sim\mathcal{CN}(0,\sum_{i=1}^{N}\left\Vert \mathbf{m}[i]\right\Vert ^{2})=\mathcal{CN}(0,1),\label{eq:inner_dis}
\end{equation}
which is Gaussian distributed with mean $0$ and variance $1$. Since
the channels $\mathbf{h}_{k},\forall k\in\mathcal{K}$ are i.i.d.,
the new random variables $\mathbf{m}^{\mathrm{H}}\mathbf{h}_{k},k\in\mathcal{K}$
are i.i.d.. Their modulus $\left\Vert \mathbf{m}^{\mathrm{H}}\mathbf{h}_{k}\right\Vert ,\forall k\in\mathcal{K}$
are also i.i.d., which are Rayleigh distributed. By defining a new
variable $Z_{k}=\left\Vert \mathbf{m}^{\mathrm{H}}\mathbf{h}_{k}\right\Vert $,
the probability density function (PDF) is expressed as
\begin{equation}
f(z_{k})=2z_{k}e^{-z_{k}^{2}},z_{k}\geq0,
\end{equation}
whose cumulative density probability (CDF) is given by
\begin{equation}
\mathbb{P}\left(Z_{k}<z_{k}\right)=1-e^{-z_{k}^{2}},z_{k}\geq0.\label{eq:CDF}
\end{equation}

\subsection{MSE Minimization}

For the problem of MSE minimization, the device subset is obtained by arranging the equivalent channel gain based on the sampled aggregate beamforming vector from a complex unit sphere. In this case, the obtained MSE performance
given $\mathbf{m}$ is written as
\begin{equation}
\mathrm{\frac{MSE}{\sigma^{2}}}=\underset{\mathcal{S}\subseteq\mathcal{K}}{\min}\enskip\underset{k\in\mathcal{S}}{\max}\frac{1}{P\left\Vert \mathbf{m}^{\mathrm{H}}\mathbf{h}_{k}\right\Vert ^{2}}.
\end{equation}
Defining a new random variable $Y_{k}=\left(P\left\Vert \mathbf{m}^{\mathrm{H}}\mathbf{h}_{k}\right\Vert ^{2}\right)^{-1}$,
we then have $Y_{k}=\left(PZ_{k}^{2}\right)^{-1}$, which is i.i.d..
According to (\ref{eq:CDF}), the CDF of $Y_{k}$ is obtained as
\begin{equation}
F(y_{k})  =\mathbb{P}\left(Y_{k}<y_{k}\right)=\mathbb{P}\left(Z_{k}>\left(\sqrt{Py_{k}}\right)^{-1}\right)=e^{-\left(Py_{k}\right)^{-1}},y_{k}>0.\label{eq:CDF1}
\end{equation}
Thereby, the PDF of variable $Y_{k}$ is obtained by deriving the
derivative of (\ref{eq:CDF1}), given by
\begin{equation}
f\left(y_{k}\right)=\left(Py_{k}^{2}\right)^{-1}e^{-\left(Py_{k}\right)^{-1}},y_{k}>0.\label{eq:pdf1}
\end{equation}
As mentioned, variables $Y_{1},Y_{2},\ldots,Y_{K}$ for all devices
are independent with the same distribution, which can be considered
as a sequence of variables sampled from the absolutely continuous
population with PDF $f(y)$ in (\ref{eq:pdf1}) and CDF $F\left(y\right)$
in (\ref{eq:CDF1}).

Define another variable $X$ as the objective value of problem (\ref{eq:SIMO1_1}),
i.e.,
\begin{equation}
X=\underset{\mathcal{S}\subseteq\mathcal{K}}{\min}\enskip\underset{k\in\mathcal{S}}{\max}\frac{1}{P\left\Vert \mathbf{m}^{\mathrm{H}}\mathbf{h}_{k}\right\Vert ^{2}}=\underset{\mathcal{S}\subseteq\mathcal{K}}{\min}\enskip\underset{k\in\mathcal{S}} {\max}Y_{k}.
\end{equation}
Let $Y_{1:K}\leq Y_{2:K}\leq\cdots\leq Y_{K:K}$ be the order statistics
obtained by arranging the variables $Y_{1},Y_{2},\ldots,Y_{K}$ in
an ascending order. We then have $X=Y_{S:K}$. The CDF of $X$ can
be derived without much difficulty such that
\begin{align}
 & G\left(x\right)=\mathbb{P}\left(X<x\right)=\mathbb{P}\left(Y_{S:K}<x\right)=1-\mathbb{P}\left(Y_{S:K}>x\right)\nonumber \\
 & =1-\mathbb{P}\left(\mbox{at most }S-1\mbox{ of }Y_{1},Y_{2},\ldots,Y_{K}\mbox{ are at most }x\right)\nonumber \\
 & =1-\sum_{s=0}^{S-1}\mathbb{P}\left(\mbox{exactly }s\mbox{ of }Y_{1},Y_{2},\ldots,Y_{K}\mbox{ are at most }x\right)\nonumber \\
 & =1-\sum_{s=0}^{S-1}\binom{K}{s}\left[\mathbb{P}\left(Y<x\right)\right]^{s}\left[\mathbb{P}\left(Y>x\right)\right]^{K-s}\nonumber \\
 & =1-\sum_{s=0}^{S-1}\binom{K}{s}e^{-s\left(Px\right)^{-1}}\left[1-e^{-\left(Px\right)^{-1}}\right]^{K-s}\label{eq:CDF2}
\end{align}

In a large-scale distributed system, there are a large population
of devices involving in the implementation of edge intelligence. Thereby,
it is also important to analyze the obtained MSE performance when $K$
is large. The following lemma establishes the asymptotic distribution
of a central order statistic when $K\rightarrow\infty$.
\begin{lem}
\label{thm:infinity_K}Denoting $q=S/K$, as $K\rightarrow\infty$,
we have
\[
\sqrt{K}f\left(F^{-1}\left(q\right)\right)\frac{\left(X-F^{-1}\left(q\right)\right)}{\sqrt{q\left(1-q\right)}}\stackrel{d}{\rightarrow}\mathcal{N}\left(0,1\right),
\]
where $F^{-1}\left(\cdot\right)$ is the inverse function of CDF $F$
in (\ref{eq:CDF1}), $\stackrel{d}{\rightarrow}$ means convergence
in distribution.
\end{lem}
\begin{proof}
According to (\ref{eq:CDF1}) and (\ref{eq:pdf1}), both CDF $F$
and PDF $f$ are continuous in their domains. Since $0<S<K$, we have
$q\in\left(0,1\right)$ and then have $F^{-1}\left(q\right)=-\frac{1}{P\ln q}>0$.
Therefore, the condition for Theorem 8.5.1 in \cite{order_statistics}
is satisfied\footnote{Let $X_{i:n}$ be the $i$-th order statistic from $n$ random variables.
The theorem is elaborated such that for $0<p<1$, let $F$ be absolutely
continuous with PDF $f$ which is positive at $F^{-1}\left(p\right)$
and is continuous at that point. For $i=\left\lfloor np\right\rfloor +1$,
as $n\rightarrow\infty$, we have $\sqrt{n}f\left(F^{-1}\left(p\right)\right)\frac{\left(X_{i:n}-F^{-1}\left(p\right)\right)}{\sqrt{p\left(1-p\right)}}\stackrel{d}{\rightarrow}\mathcal{N}\left(0,1\right)$.}. This completes the proof.
\end{proof}
Lemma \ref{thm:infinity_K} shows that random variable $X$ is asymptotically
normal after suitable normalization. It also indicates that the expectation
and the variance of MSE performance $X$ can be approximated with
probability $q=S/K$, expressed by

\begin{equation}
\mathbb{E}\left(X\right)\simeq F^{-1}\left(q\right),\label{eq:appro_mu}
\end{equation}

\begin{equation}
\mathrm{VAR}\left(X\right)\simeq\frac{q\left(1-q\right)}{K\left[f\left(F^{-1}\left(q\right)\right)\right]^{2}}=\frac{\left(1-q\right)}{KP^{2}q\left(\ln q\right)^{4}},\label{eq:appro_var}
\end{equation}
where $F^{-1}\left(q\right)$ and $f\left(F^{-1}\left(q\right)\right)$
are obtained from (\ref{eq:CDF1}) and (\ref{eq:pdf1}) such that
\[
F^{-1}\left(q\right)=-\frac{1}{P\ln q},\:f\left(F^{-1}\left(q\right)\right)=Pq\left(\ln q\right)^{2}.
\]

\begin{thm}
\label{thm:optimal}The proposed random aggregate beamforming-based
method for problem (\ref{eq:SIMO1_1}) approaches the optimal solution
when $K\rightarrow\infty$ and $S\ll K$.
\end{thm}
\begin{proof}
For the number of devices $K\rightarrow\infty$ and a fixed number
$S\ll K$ , we have $q=S/K\rightarrow0$. The limit of the right terms
of (\ref{eq:appro_mu}) and (\ref{eq:appro_var}) are respectively
given by
\begin{equation}
\lim_{q\rightarrow0}F^{-1}\left(q\right)=\lim_{q\rightarrow0}-\frac{1}{P\ln q}=0.
\end{equation}
Thus, we have $\mathbb{E}\left(X\right)\rightarrow0$ for $K\rightarrow\infty$
and $S\ll K$.

Denoting $X^{*}$ as the optimal objective of problem (\ref{eq:SIMO1}),
we have $X\geq X^{*}>0$. Then the expectations of $X$ and $X^{*}$
satisfy:
\begin{equation}
\mathbb{E}\left(X\right)\geq\mathbb{E}\left(X^{*}\right)>0.
\end{equation}
From the above observations, we have $\mathbb{E}\left(X^{*}\right)\rightarrow0$
and further arrive at $\mathbb{E}\left(X\right)=\mathbb{E}\left(X^{*}\right)$,
which indicates that the proposed random aggregate beamforming-based
method approaches the optimal solution for $K\rightarrow\infty,S\ll K$.
\end{proof}

\subsection{Maximum Number of Selected Devices}

Similar to the previous subsection, the solution of the aggregate
beamforming vector $\mathbf{m}$ to problem (\ref{eq:SIMO3})
is determined by randomly sampling a vector from $\mathcal{M}$.
Then the devices that guarantee the MSE constraint are selected, i.e.,
\begin{align}
\mathcal{S} & =\left\{ k|\mathrm{MSE}_{k}\leq\overline{\mathrm{MSE}},\forall k\in\mathcal{K}\right\} \nonumber\\
 & =\left\{ k|\frac{1}{P\left\Vert \mathbf{m}^{\mathrm{H}}\mathbf{h}_{k}\right\Vert ^{2}}\leq\frac{\overline{\mathrm{MSE}}}{\sigma^{2}},\forall k\in\mathcal{K}\right\} \nonumber\\
 & =\left\{ k|Y_{k}\leq\tilde{x},\forall k\in\mathcal{K}\right\} ,
\end{align}
where $\tilde{x}=\text{\ensuremath{\overline{\mathrm{MSE}}}/\ensuremath{\sigma^{2}}}$.
Defining the discrete random variable $\Lambda=\left|\mathcal{S}\right|$,
the probability mass function (PMF) is obtained as
\begin{align}
\mathbb{P}\left(\varLambda=S\right) & =\mathbb{P}\left(\mbox{Exactly }S\mbox{ of }Y_{1},Y_{2},\ldots,Y_{K}\mbox{ are less than }\tilde{x}\right)\nonumber \\
 & =\binom{K}{S}\prod_{i=1}^{S}\mathbb{P}\left(Y<\tilde{x}\right)\prod_{i=1}^{K-S}\mathbb{P}\left(Y>\tilde{x}\right)\nonumber \\
 & =\binom{K}{S}e^{-S\left(P\tilde{x}\right)^{-1}}\left(1-e^{-\left(P\tilde{x}\right)^{-1}}\right)^{K-S}.\label{eq:PFM}
\end{align}

In a large-scale system where there are massive devices, we arrive
at the following theorem.
\begin{thm}
\label{thm:MSE_thre}Denote $\Lambda_{\min}=\frac{K}{e^{1/\left(P(\tilde{x}-3\sigma)\right)}}$,
$\Lambda_{\max}=\frac{K}{e^{1/\left(P(\tilde{x}+3\sigma)\right)}}$,
where $\sigma=\sqrt{\frac{1-\exp\left(-\frac{1}{P\tilde{x}}\right)}{KP^{2}\exp\left(-\frac{1}{P\tilde{x}}\right)\left(-\frac{1}{P\tilde{x}}\right)^{4}}}$.
The random variable $\Lambda$ is approximately symmetric within the
interval of $[\Lambda_{\min},\Lambda_{\max}]$ with the expectation
of \textup{$\mathbb{E}\left(\Lambda\right)=Ke^{-1/\left(P\tilde{x}\right)}$}
when $K\rightarrow\infty$.

\end{thm}
\begin{proof}
When the expectation of the MSE performance in the MSE minimization
problem equals to $\tilde{x}$ for $K\rightarrow\infty$, we have
\begin{align}
\mathbb{E}\left(X\right)\simeq  & F^{-1}\left(q\right)=-\frac{1}{P\ln q}=\tilde{x}\nonumber \\
\Leftrightarrow & -\frac{1}{P\ln q}=\tilde{x}\nonumber
\Leftrightarrow  q=\exp\left(-\frac{1}{P\tilde{x}}\right).
\end{align}
Then the variance of $X$ can be expressed as
\begin{equation}
\mathrm{VAR}\left(X\right)\simeq\frac{q\left(1-q\right)}{K\left[f\left(F^{-1}\left(q\right)\right)\right]^{2}}=\frac{1-\exp\left(-\frac{1}{P\tilde{x}}\right)}{KP^{2}\exp\left(-\frac{1}{P\tilde{x}}\right)\left(-\frac{1}{P\tilde{x}}\right)^{4}},
\end{equation}
which approaches $0$ when $K\rightarrow\infty$. Since the random
variable $X$ is asymptotically normal, variable $X$ is within the
interval of $\left[\tilde{x}-3\sigma,\tilde{x}+3\sigma\right]$ where
$\sigma=\sqrt{\frac{1-\exp\left(-\frac{1}{P\tilde{x}}\right)}{KP^{2}\exp\left(-\frac{1}{P\tilde{x}}\right)\left(-\frac{1}{P\tilde{x}}\right)^{4}}}$
is the standard deviation.

For any MSE performance variable $X'$ whose expectation $\mathbb{E}\left(X'\right)$
is within the interval of $\in\left[\tilde{x}-3\sigma,\tilde{x}+3\sigma\right]$,
its variance $\mathrm{VAR}\left(X'\right)$ approximately equals
$\mathrm{VAR}\left(X\right)$ owing to the small standard deviation
$\sigma$. Similarly, the probability $q'$ for variable $X'$
is written as $q=\exp\left(-\frac{1}{P\mathbb{E}\left(X'\right)}\right)$.
Thus, the corresponding number of the selected devices is $S=\left\lfloor \frac{K}{e^{1/\left(P\mathbb{E}\left(X'\right)\right)}}\right\rfloor $.
In the following, we consider two cases where $\mathbb{E}\left(X'\right)$
is not in the interval of $\left[\tilde{x}-3\sigma,\tilde{x}+3\sigma\right]$.

Case1: When $\mathbb{E}\left(X'\right)<\tilde{x}-3\sigma$, we have
$S<\frac{K}{e^{1/\left(P(\tilde{x}-3\sigma)\right)}}$ and $X'$ is
approximately within $\left[\mathbb{E}\left(X'\right)-3\sigma,\mathbb{E}\left(X'\right)+3\sigma\right]$.
Since $\mathbb{E}\left(X'\right)+3\sigma<\tilde{x}$, it shows that
selecting $S$ devices guarantees the MSE constraint, which further
suggests that more devices can be selected.

Case2: When $\mathbb{E}\left(X'\right)>\tilde{x}+3\sigma$, we have
$S>\frac{K}{e^{1/\left(P(\tilde{x}+3\sigma)\right)}}$ and $X'$ is
approximately within $\left[\mathbb{E}\left(X'\right)-3\sigma,\mathbb{E}\left(X'\right)+3\sigma\right]$.
Since $\mathbb{E}\left(X'\right)-3\sigma>\tilde{x}$, the MSE constraint
cannot be guaranteed if selecting $S>\frac{K}{e^{1/\left(P(\tilde{x}+3\sigma)\right)}}$
devices.

Thus, the approximate minimum and maximum number of the selected
devices under the random aggregate beamforming based method are respectively
\begin{equation}
\Lambda_{\min}=\frac{K}{e^{1/\left(P(\tilde{x}-3\sigma)\right)}},\quad\Lambda_{\mathrm{max}}=\frac{K}{e^{1/\left(P(\tilde{x}+3\sigma)\right)}}.
\end{equation}
As stated above, when the number of the selected devices $S\in[\Lambda_{\min},\Lambda_{\max}]$,
the obtained MSE performance has almost the same variance. Therefore,
the PMF of variable $\Lambda$ is symmetric within the interval of
$[\Lambda_{\min},\Lambda_{\max}]$ with the expectation of $\mathbb{E}\left(\Lambda\right)=Ke^{-1/\left(P\tilde{x}\right)}$.
This completes the proof.
\end{proof}

\section{Refined Method}

The previous section first gives the proposed random aggregate beamforming-based scheme,
and then presents theoretical analysis in terms of the two objectives
when $K$ is large. This section focuses on the performance improvement
of the proposed methods when $K\ll\infty$. Specifically, we refine
the proposed scheme by randomizing the vector $\mathbf{m}$ from
$\mathcal{M}$ for $N_{m}$ times. The solutions with the best performance
are obtained among $N_{m}$ aggregate beamforming vectors. The refined methods
for both MSE minimization and the number of selected devices maximization
are detailed in Algorithm \ref{alg:Random} and Algorithm \ref{alg:DeviceNumber}.
Additionally, the effectiveness of the refined methods are analyzed.
\begin{algorithm}[t]
	\small	
\caption{Random Aggregate Beamforming-based Design for MSE Minimization.}
	\textbf{Input}:  $N_m$, $S$ \quad	\textbf{Output}: $\mathcal{S}$, $\mathbf{m}$\quad
	\textbf{Initialize}: $\mathrm{val}=0$, $\mathrm{tmp_{max}}=0$
\begin{algorithmic}[1]	
\State \textbf{for} $n=1:N_m$
	\State \quad Sample $\mathbf{m}_r$ from $\mathcal{CN}(0,\mathbf{I})$
	\State \quad Normalize $\mathbf{m}_r=\mathbf{m}_r/\left\Vert \mathbf{m}_r\right\Vert $
	\State \quad Calculate $\mathrm{tmp_{k}}=\left\Vert \mathbf{m}_r^{\mathrm{H}}\mathbf{h}_{k}\right\Vert^2, \forall k \in \mathcal{K}$
	\State  \quad Find the $S$-th largest $\mathrm{tmp_{k}},\forall k \in \mathcal{K}$
	\State \quad If $\mathrm{tmp}>\mathrm{tmp_{max}}$.
	\State \quad \quad $\mathcal{S}=\left\{ k\:\vert\left\Vert \mathbf{m}_{r}^{\mathrm{H}}\mathbf{h}_{k}\right\Vert ^{2}\geq\mathrm{tmp_{max}},\forall k\in\mathcal{K}\right\} $
	\State \quad \quad $\mathbf{m}=\mathbf{m}_r$.
\end{algorithmic}
\label{alg:Random} 
\end{algorithm}

\begin{algorithm}[t]
	\small	
\caption{Random Aggregate Beamforming-based Design for the Number of Selected Devices Minimization.}
\textbf{Input}:$N_m$, $\overline{\mathrm{MSE}}/{\sigma^{2}}$\quad
\textbf{Output}: $\mathcal{S}_\mathrm{max}$, $\mathbf{m}$\quad
\textbf{Initialize}: $\mathrm{S_{max}}=0$
\begin{algorithmic}[1]
\State \textbf{for} $n=1:N_m$
\State \quad Sample $\mathbf{m}_r$ from $\mathcal{CN}(0,\mathbf{I})$
\State \quad Normalize $\mathbf{m}_r=\mathbf{m}_r/\left\Vert \mathbf{m}_r\right\Vert $
\State \quad Calculate $\mathrm{tmp}_{k}=\left(P\left\Vert \mathbf{m}_{r}^{\mathrm{H}}\mathbf{h}_{k}\right\Vert ^{2}\right)^{-1}, \forall k \in \mathcal{K}$
\State \quad $\mathcal{S}=\left\{k\,\vert\mathrm{tmp_{k}}\leq \overline{\mathrm{MSE}}/{\sigma^{2}},\forall k\in\mathcal{K}\right\}$
\State \quad If $\left|\mathcal{S}\right|>\left|\mathcal{S}_\mathrm{max}\right|$
\State \quad \quad $\mathcal{S}_\mathrm{max}=\mathcal{S}$
\State \quad \quad $\mathbf{m} = \mathbf{m}_r$.
\end{algorithmic}
\label{alg:DeviceNumber}
\end{algorithm}

Denote $\mathbf{m}^{*}$ and $\mathcal{S}^{*}$ as the the optimal
aggregate beamforming vector and selected device subset, respectively. According
to Lemma \ref{thm:infinity}, there are infinite optimal solutions,
i.e., $\mathbf{m}^{*}e^{j\theta},\forall\theta\in\mathcal{R}$,
with the same optimal objectives of the MSE and the number
of selected devices, which are denoted as $x^{*}$ and $\gamma^{*}$
respectively. As the aggregate beamforming vector $\mathbf{m}$
is uniformly sampled from the complex unit sphere $\mathcal{M}$,
the probability that the sampled vector $\mathbf{m}$ is optimal,
i.e., $\mathbb{P}(\mathbf{m}=\mathbf{m}^{*})$, can be interpreted
as the proportion of the region of the optimal solutions on the whole
sphere $\mathcal{M}$. Despite infinite optimal solutions of $\mathbf{m}$,
the probability can be very small. Nevertheless, we can model the probability of a sub-optimal solutions
of both problems in an implicit manner, which are detailed in the next subsections.

\subsection{MSE Minimization}

For problem (\ref{eq:SIMO1_1}), the optimal objective as $x^{*}$
is related to the channel of all devices $\mathbf{h}_{k},\forall k\in\mathcal{K}$.
Let $\varDelta\in\mathbb{R}^{+}$ be the difference between the optimal
objective and the sub-optimal one. Then, the probability that the
obtained sub-optimal objective under the proposed method less than a given value
$x$ where $x>x^{*}$ can be modelled as
\begin{equation}
\varsigma=\mathbb{P}\left(X<x;\left\{ \mathbf{h}_{k},\forall k\in\mathcal{K}\right\} \right)=\mathbb{P}\left(X<x^{*}+\varDelta\right).\label{eq:MSE_mulple}
\end{equation}

The optimal objective $x^{*}$ is difficult to obtain as the problem
(\ref{eq:SIMO1_1}) is non-convex. In addition, it is difficult to
obtain the distribution of objective $X$ under $\mathbf{h}_{k},\forall k\in\mathcal{K}$,
where $\mathbf{m}$ is uniformly distributed on the unit sphere.
Thus, the probability $\varsigma$ cannot be analytically expressed.
Instead, we model $\sigma$ as an implicit function of $x^{*}$ and
$\varDelta$, i.e., $\varsigma=H\left(\varDelta;x^{*}\right)$ with
the range of $\left[0,1\right]$. The function is monotonically nondecreasing
w.r.t. the difference $\varDelta>0$, which is useful for our later
analysis. To illustrate the distribution, we search the feasible solutions
to obtain $H\left(\varDelta;x^{*}\right)$ under one channel realization
in the scenario with $K=50$ devices and an aggregator equipped with
$N=2$ antennas. As can be seen in Fig. \ref{fig:MSE_probability},
for the optimal objective $x^{*}=0.455$, the probability $\varsigma$
is 0.01 when the difference $\varDelta=0.011$. By randomizing the
vector $\mathbf{m}$ from $\mathcal{M}=\left\{ \mathbf{m}|\text{\ensuremath{\left\Vert \mathbf{m}\right\Vert }=1},\mathbf{m}\in\mathbb{C}^{2}\right\} $
for $100$ times, we get $63.4\%$ chance that the obtanied objective
$x$ is whithin the interval of $\left[0.455,0.466\right]$.

\begin{figure}[tbh]
\begin{tabular}{cc}
\begin{minipage}[t]{0.5\textwidth} \includegraphics[width=1\columnwidth]{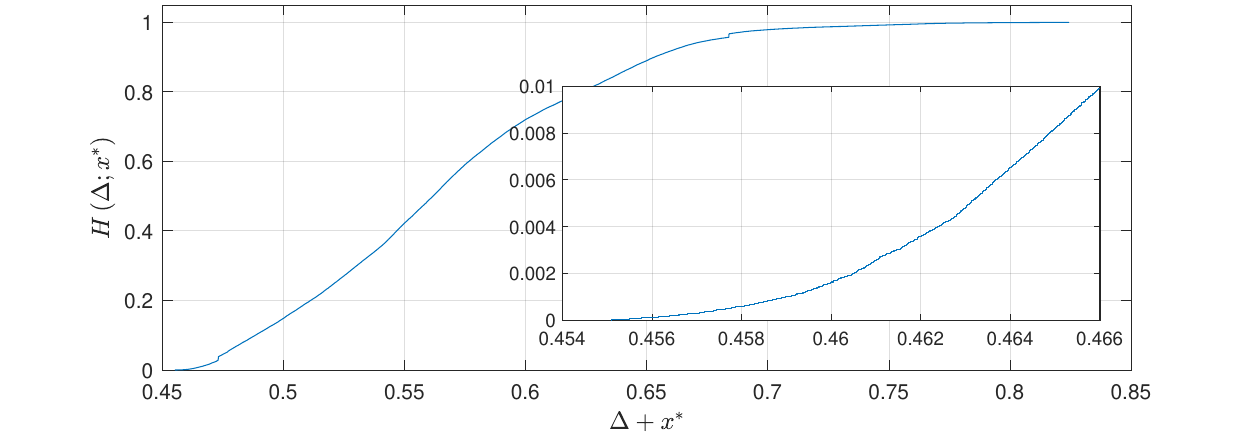}\captionsetup{font={small}}\caption{Distribution $H\left(\varDelta;x^{*}\right)$ under one channel realization with $S=10$ devices are selected.}
\label{fig:MSE_probability}
\end{minipage}
\begin{minipage}[t]{0.5\textwidth} \includegraphics[width=1\columnwidth]{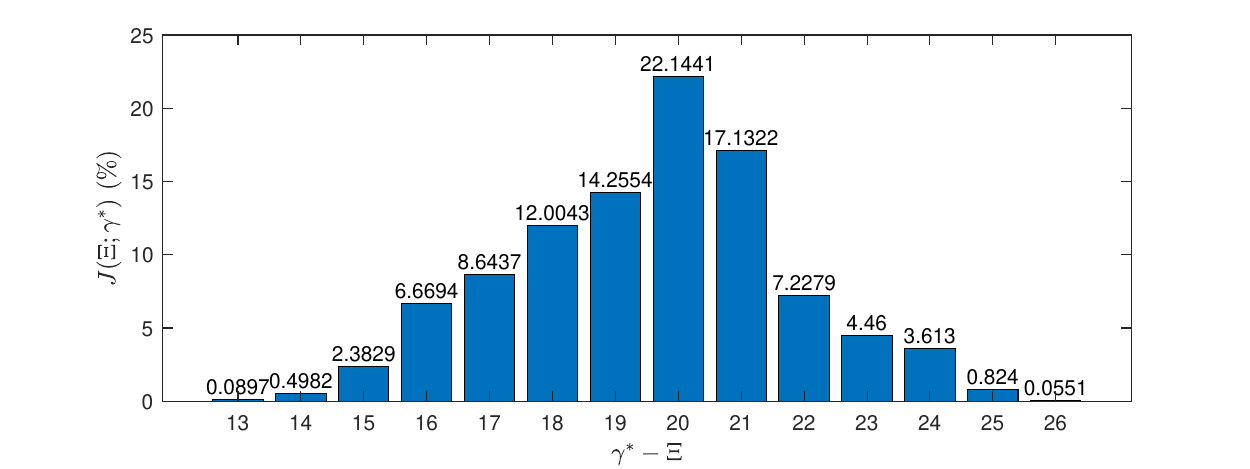}\captionsetup{font={small}}\caption{PMF function $\beta=J\left(\Xi;\gamma^{*}\right)$ under one channel realization with the threshold $\overline{\mathrm{MSE}}/\sigma^{2}=1$.}
\label{fig:Num_probability}
\end{minipage}         
\end{tabular}
\end{figure}


Considering $N_{m}$ randomizations of the aggregate vectors, we arrive
at the following theorem.
\begin{thm}
\label{thm:random_ana} Let $\bar{X}=\min\left(X_{1},\ldots,X_{i},\ldots,X_{N_{m}}\right)$
be the obtained objective value given $N_{m}$ random aggregate beamforming
vectors under the channel conditions of $\mathbf{h}_{k},\forall k\in\mathcal{K}$.
For any objective value difference \textup{$\varDelta>0$,} we have
\[
\lim_{N_{m}\rightarrow\infty}\mathbb{P}\left(\bar{X}<x^{*}+\varDelta\right)=1.
\]
\end{thm}
\begin{proof}
Since $N_{m}$ random aggregate beamforming vectors are uniformly
sampled from the same complex unit sphere $\mathcal{M}$, the probability
of $\bar{X}$ to be less than $x=x^{*}+\varDelta$ is given by
\begin{align*}
\mathbb{P}\left(\bar{X}<x\right) & =\mathbb{P}\left(\bar{X}<x^{*}+\varDelta\right)\\
 & =1-\mathbb{P}\left(\bar{X}>x^{*}+\varDelta\right)\\
 & =1-\prod_{i=1}^{N_{m}}\mathbb{P}\left(X_{i}>x^{*}+\varDelta\right)\\
 & =1-\left(1-H\left(\varDelta;x^{*}\right)\right)^{N_{m}},
\end{align*}
where $H\left(\cdot\right)$ is a monotonic function w.r.t the difference
$\varDelta$.  Given $\varDelta>0$, for any $\delta>0$, the following
equivalence holds, i.e.,
\begin{align*}
 & \left|\mathbb{P}\left(\bar{X}<x^{*}+\varDelta\right)-1\right|<\delta\\
\Leftrightarrow & \left(1-H\left(\varDelta;x^{*}\right)\right)^{N_{m}}<\delta\\
\Leftrightarrow & N_{m}>\left\lfloor \log_{1-H\left(\varDelta;x^{*}\right)}\delta\right\rfloor ,
\end{align*}
where $\left\lfloor \cdot\right\rfloor $ is the floor function. Therefore,
given $\varDelta>0$, for $\forall\delta>0$, there exist a number
$N=\left\lfloor \log_{1-H\left(\varDelta;x^{*}\right)}\delta\right\rfloor $
such that when $N_{m}>N$, we have $\left|\mathbb{P}\left(\bar{x}<x^{*}+\varDelta\right)-1\right|<\delta$.
This completes the proof.
\end{proof}

\subsection{Maximizing The Number of Selected Devices}

As for the mixed combinatorial problem (\ref{eq:SIMO3}), the optimal
objective $\gamma^{*}$ is difficult to obtain, which does not have
the close-form expression w.r.t. $\mathbf{h}_{k},\forall k\in\mathcal{K}$.
Also, the PMF of the number of selected devices $\Lambda$ obtained
by the proposed method is unknown under $\mathbf{h}_{k},\forall k\in\mathcal{K}$.
Thus, the probability that the number of selected devices is greater
than a sub-optimal objection $\gamma$ cannot be explicitly expressed,
i.e.,
\begin{equation}
\beta=\mathbb{P}\left(\Lambda=\gamma\right)=\mathbb{P}\left(\Lambda=\gamma^{*}-\Xi\right),
\end{equation}
where $\Xi\in\mathbb{N}$ is the difference between the optimal number
of selected devices and the sub-optimal one. It is a function of $\gamma^{*}$
and $\Xi$, i.e., $\beta=J\left(\Xi;\gamma^{*}\right)$, which is
monotonically nondecreasing w.r.t. $\Xi$ within the range of $\left[0,1\right]$.
Likewise, we illustrate the probability $J\left(\Xi;\gamma^{*}\right)$
by searching the feasible solutions under one channel realization
in a scenario with $K=50$ devices and an aggregator equipped with
$N=2$ antennas. As can be seen in Fig. \ref{fig:Num_probability},
the optimal objective $\gamma^{*}=26$, the probability of which
is $0.0551\%$. In this case, $1000$ times of randomization can lead
to $42.37\%$ chance to obtain the optimal solution.


To improve the performance, $N_{m}$ randomizations of the aggregate
vectors are considered. The following Theorem can be arrived.
\begin{thm}
\label{thm:random_ana-1} Let $\bar{\varLambda}=\min\left(\varLambda_{1},\cdots,\varLambda_{i},\cdots,\varLambda_{N_{m}}\right)$
be the obtained objective value via the refined method, where $\varLambda_{i}$
is the objective value under the $i$-th random vector $\mathbf{m}_{i}$.
For any objective value difference \textup{$\Xi\in\mathbb{N}$,} we
have
\[
\lim_{N_{m}\rightarrow\infty}\mathbb{P}\left(\bar{\varLambda}>\gamma^{*}-\Xi\right)=1.
\]
\end{thm}
\begin{proof}
The proof is similar to that of Theorem \ref{thm:random_ana}.
The probability of $\bar{\varLambda}$ to be greater than $\gamma=\gamma^{*}-\Xi$
is implicitly expressed as $\mathbb{P}\left(\bar{\varLambda}>\gamma\right)=\mathbb{P}\left(\bar{\varLambda}>\gamma^{*}-\Xi\right)=1-\left(1-J\left(\Xi;\gamma^{*}\right)\right)^{N_{m}}.$ Given
$\Xi\in\mathbb{N}$, for $\forall\delta>0$, there exists the number
$N=\left\lfloor \log_{1-J\left(\Xi;\gamma^{*}\right)}\delta\right\rfloor $
such that when $N_{m}>N$, we have $\left|\mathbb{P}\left(\bar{\varLambda}>\gamma^{*}-\Xi\right)-1\right|<\delta$.
This completes the proof.
\end{proof}

\subsection{Observations}

When $K\ll\infty$, we refine the proposed methods in Section IV by
randomizing $N_{m}$ aggregate beamforming vectors aiming at the performance
improvement. The proposed methods in Section IV have a computational
complexity of $\mathcal{O}\left(K\right)$, and thereby the required computational
complexity for the refined methods is $\mathcal{O}\left(N_{m}K\right)$.
 From the above analysis, the gaps between the sub-optimal performance
and the optimal ones are less than $\varDelta$ and $\Xi$ when $N_{m}>\left\lfloor \log_{1-H\left(\varDelta;x^{*}\right)}\delta\right\rfloor $
and $N_{m}>\left\lfloor \log_{1-J\left(\Xi;\gamma^{*}\right)}\delta\right\rfloor $
where $\delta>0$ is arbitrary small. Reversely, given the randomization
number $N_{m}$, the corresponding differences are respectively $\varDelta<H^{-1}\left(1-\sqrt[N_{m}]{\delta};x^{*}\right)$
and $\Xi<J^{-1}\left(1-\sqrt[N_{m}]{\delta};\gamma^{*}\right)$.
Despite that the inverse functions $H^{-1}\left(\cdot\right)$ and
$J^{-1}\left(\cdot\right)$ cannot be explicitly expressed, the refined
methods with multiple randomizations provide an obvious insight
of performance improvement. The reason behind is that both functions
$H$ and $J$ are monotonically nondecreasing w.r.t. $\varDelta$
and $\Xi$, which indicate that the inverse functions are also monotonically
nondecreasing w.r.t. $1-\sqrt[N_{m}]{\delta}$. Therefore, a larger
randomlization $N_{m}$ leads to a small differences $\varDelta$
and $\Xi$ from the optimal objectives.

\section{Simulation}

This section presents the simulation results to demonstrate the effectiveness
of the proposed random aggregate beamforming-based methods as well
as the refined ones. Besides, the theoretical analysis is confirmed.
The maximum transmit power of the devices is set to $P=0$ dB. To
showcase the advantage of the proposed method, we compare it to reference
methods in terms of the MSE performance and the number of the selected
devices for the investigated problems. For the problem of MSE minimization,
the benchmarks are
\begin{itemize}
\item Iterative device selection and aggregate beamforming design: Given
the selected devices, the aggregate beamforming vector is obtained
by utilizing DC method. After the aggregate beamforming vector is
obtained, arrange the equivalent channel power and select the devices
with $S$ largest values. The solutions are obtained by alternatively
solving the above two subproblems until convergence.
\item Random device selection and aggregate beamforming design: Randomly
select $S$ out of $K$ devices, and optimize the aggregate beamforming
vector by DC technique.
\end{itemize}
The problem of the number of selected devices maximization can be
modelled as a sparse problem. DC representation and $\ell_{1}-\mathrm{norm}$
techniques are the two most common methods to handle the sparsity.
The comparing benchmarks are:
\begin{itemize}
\item A novel DC algorithm \cite{kaiyang}: The DC representation is adopted
to induce the sparsity, based on which the DC technique is used to
obtain the aggregate beamforming vector.
\item $\ell_{1}$+SDR \cite{SDR}: The sparsity is induced by the $\ell_{1}-\mathrm{norm}$
technique, and the aggregate beamforming vector is obtained by SDR
method.
\end{itemize}

\subsection{Analysis Verification}

To verify the theoretical analysis in terms of the MSE and the selected
number under the proposed random aggregate beamforming vector, we
conduct 1000 Monte Carlo simulations, where the number of antennas
at the aggregator is set to $N=8$.

For the problem of MSE minimization, we consider two settings where
the fixed numbers of selected devices $S$ are $10$ and $20$ respectively.
Fig. \ref{fig:CDF} displays the theoretical and empirical CDFs under
these two settings, where the coinciding lines validate the derived
distribution of $\mathrm{MSE}/\sigma^{2}$ in (\ref{eq:CDF2}). By
comparing Fig. \ref{fig:CDF}(a) and Fig. \ref{fig:CDF}(b), we can
observe that selecting more devices leads to a larger aggregate error,
because the scaling factor $\eta$ shall be increased in order to
align signals from more devices. In addition, as the increase of the
number of devices $K$, the obtained $\mathrm{MSE}/\sigma^{2}$ together
with its variance decrease with the trends towards $0$. To confirm
the Lemma \ref{thm:infinity_K}, we compare the CDF of normal distribution
with mean $F^{-1}\left(q\right)$ and variance $\frac{\left(1-q\right)}{KP^{2}q\left(\ln q\right)^{4}}$
with the numerical one, which is displayed in Fig. \ref{fig:CDF_norm}.
It shows that the distribution of the obtained $\mathrm{MSE}/\sigma^{2}$
converges to the normal distribution as the increase of $K$, which
confirms Lemma \ref{thm:infinity_K}. Moreover, the mean and variance
of $\mathrm{MSE}/\sigma^{2}$ are shown to approach $0$ when $K$
is getting larger, which validates Theorem \ref{thm:optimal}.


\begin{figure}[t]
\begin{tabular}{cc}   
\begin{minipage}[t]{0.5\textwidth}
\centering   
\subfigure[]{\includegraphics[width=0.45\columnwidth]{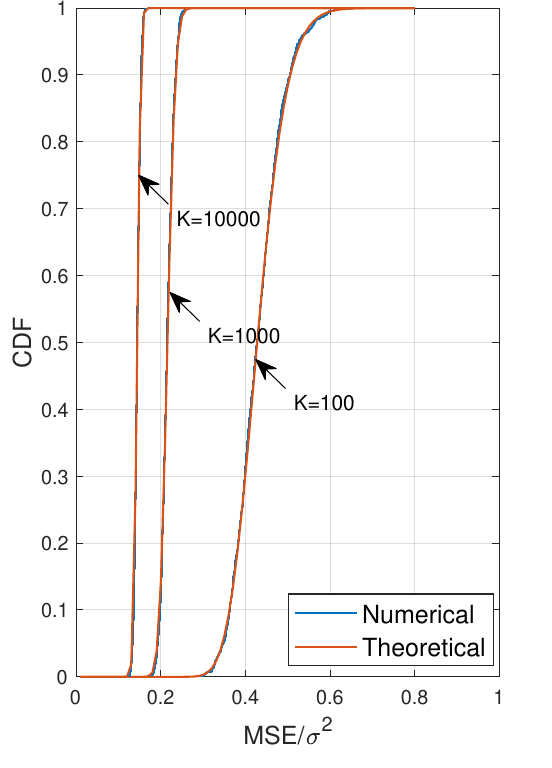}}
\subfigure[] {\includegraphics[width=0.45\columnwidth]{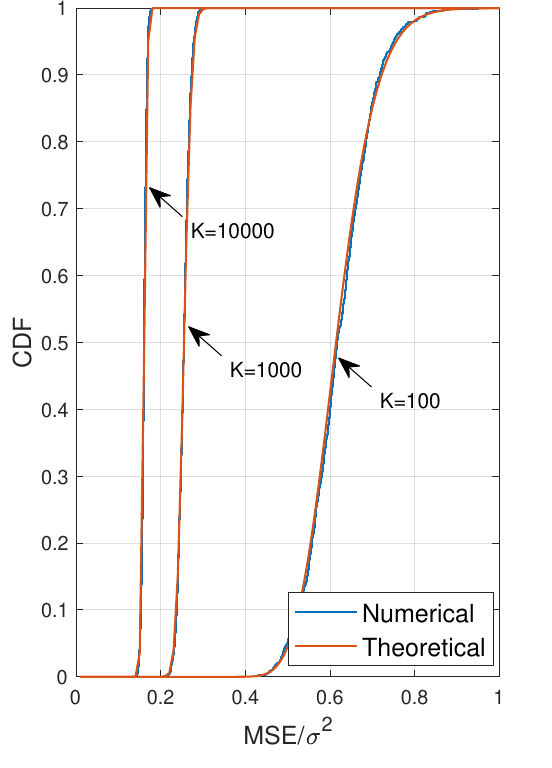}}
\captionsetup{font={small}}\caption{Derived and numerical CDFs. (a) $S=10$. (b) $S=20$.}
\label{fig:CDF}
\end{minipage}
\begin{minipage}[t]{0.5\textwidth}		
\includegraphics[width=0.96\columnwidth]{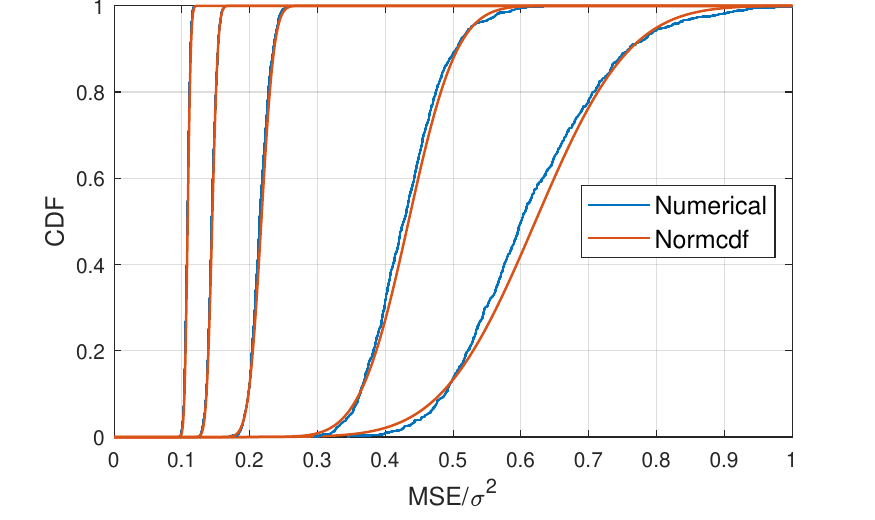}
\captionsetup{font={small}}\caption{Empirical CDF and the approximated normal distribution. The number of the devices from right to left are respectively: $K=50,K=10^2,K=10^3,K=10^4,K=10^5$.}
\label{fig:CDF_norm}
\end{minipage}
\end{tabular}
\end{figure}

%

For the problem of the maximum number of selected devices, we present
both the numerical PMF and the the derived one, where the MSE threshold
is set to $\overline{\mathrm{MSE}}/\sigma^{2}=0.2$. As is shown in Fig.
\ref{fig:PMF}, the numerical PMF has approximately the same number
of the selected devices as the derived one, which ensures the correctness
of (\ref{fig:PMF}). In addition, the PMF becomes symmetric with the
increase of $K$. To verify Theorem \ref{thm:MSE_thre}, we calculate
the average, minimum and maximum number of selected devices in the
scenarios with the threshold $\overline{\mathrm{MSE}}/\sigma^{2}$
set to $0.2$, $0.35$ and $0.5$. Table \ref{tab:Rep_accuracy} gives
the average, minimum and maximum number of the selected devices. From
the Table, we can see that the average number of the selected devices
are almost the same for both numerical and theoretical results. The
numerical minimum number is smaller than the theoretical one, while
the numerical maximum number is larger than the theoretical one. These
relative gaps are diminished as the increase of $K$, which verifies
Theorem \ref{thm:MSE_thre}.


\begin{figure}
\centering   
\subfigure[] {\includegraphics[width=0.48\columnwidth]{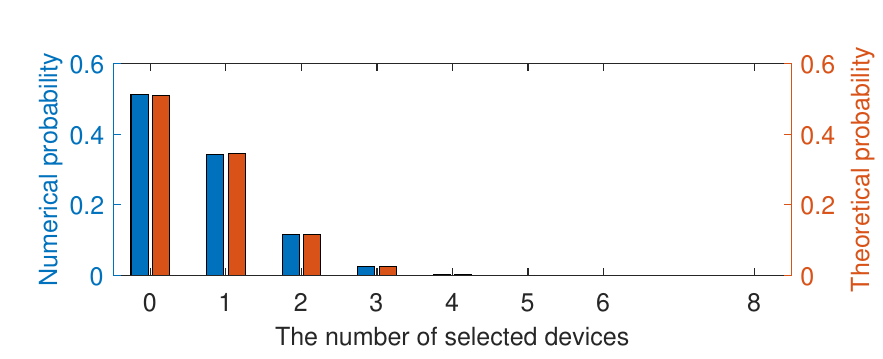}} 
\subfigure[] {\includegraphics[width=0.48\columnwidth]{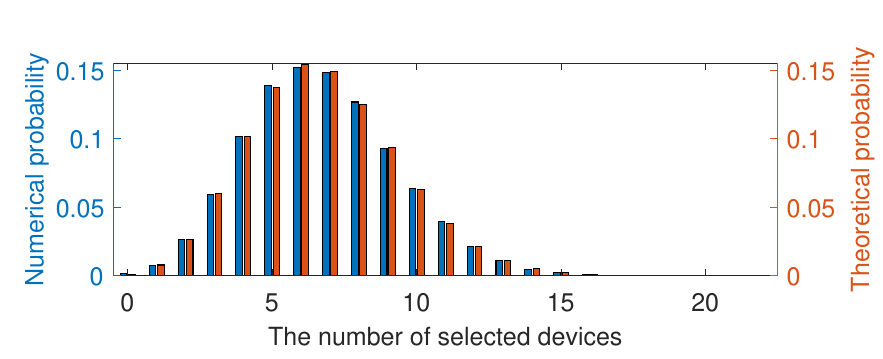}} 
\subfigure[] {\includegraphics[width=0.48\columnwidth]{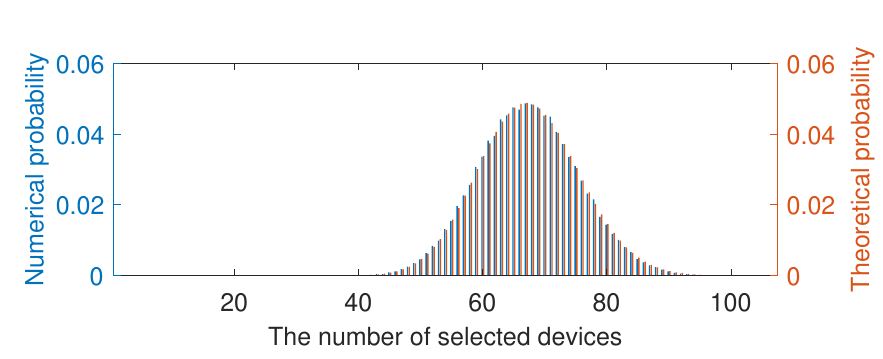}}
\captionsetup{font={small}}\caption{Numerical and theoretical PMFs. (a) $K=10^2$. (b) $K=10^3$. (c) $K=10^4$.}
\label{fig:PMF}
\end{figure}

\begin{table}[t]
\centering
\begin{tabular}{>{\centering}m{1.2cm}|>{\centering}m{1.4cm}|>{\centering}m{1cm}|>{\centering}m{1cm}|>{\centering}m{1cm}|>{\centering}m{1cm}|>{\centering}m{1cm}|>{\centering}m{1cm}|>{\centering}m{1cm}|>{\centering}m{1cm}|>{\centering}m{1cm}}
\hline
\multirow{2}{2.5cm}{Scenario} & Threshold $\overline{\mathrm{MSE}}/\sigma^{2}$ & \multicolumn{3}{c|}{0.2} & \multicolumn{3}{c|}{0.35} & \multicolumn{3}{c}{0.5}\tabularnewline
\cline{2-2} \cline{3-3} \cline{4-4} \cline{5-5} \cline{6-6} \cline{7-7} \cline{8-8} \cline{9-9} \cline{10-10} \cline{11-11}
 & Stochastic & Average &Minimum &Maximum & Average & Minimum & Maximum & Average & Minimum & Maximum\tabularnewline
\hline
\multirow{2}{2.5cm}{ $K=10^4$} & Numerical & 67.36 & 37 & 103 & 574.21 & 479 & 656 & 1353.49 & 1234 & 1487\tabularnewline
\cline{2-2} \cline{3-3} \cline{4-4} \cline{5-5} \cline{6-6} \cline{7-7} \cline{8-8} \cline{9-9} \cline{10-10} \cline{11-11}
 & Theoretical & 67.38 & 52.20 & 84.94 & 574.33 & 528.38 & 621.41 & 1353.35 & 1284.95 & 1421.76\tabularnewline
\hline
\multirow{2}{2.5cm}{ $K=10^5$} & Numerical & 674.45 & 585 & 778 & 5744.21 & 5546 & 6046 & 13533.16 & 13133 & 13908\tabularnewline
\cline{2-2} \cline{3-3} \cline{4-4} \cline{5-5} \cline{6-6} \cline{7-7} \cline{8-8} \cline{9-9} \cline{10-10} \cline{11-11}
 & Theoretical & 673.79 & 598.86 & 754.09 & 5743.26 & 5596.68 & 5890.97 & 13533.53 & 13317.18 & 13749.87\tabularnewline
\hline
\end{tabular}\captionsetup{font={small}}\caption{Average, minimum and maximum values of the number of the selected
devices.}
\label{tab:Rep_accuracy}
\end{table}

\subsection{MSE Performance}

In order to showcase the performance of the proposed random aggregate
beamforming-based design, we conduct simulations for 100 channel realizations.
The proposed methods are compared to the benchmarks previously detailed
in the scenarios with different number of antennas $N$ at the aggregator.
We label the iterative device selection and aggregate beamforming
design and the random device selection and aggregate beamforming design
as `Benchmark1' and `Benchmark2' respectively. As shown in Fig. \ref{fig:MSE_Ant}(a),
MSE performance $\mathrm{MSE}/\sigma^{2}$ obtained by both benchmarks
are decreasing as the increase of $N$. This is due to the fact that
more deployment of antennas provides more degree of freedom to align
the signals from the selected devices. Obtaining the aggregate beamforming
vector via the optimization methods can enjoy the advantage of multiple
antennas. In contrast, sample a vector from the complex unit sphere
as the aggregate beamforming vector cannot exploit the merit of multiple
antennas. Hence, the obtained MSE performance under the proposed method
remains unchanged w.r.t. the number of antennas $N$. When $K$ grows,
the obtained $\mathrm{MSE}/\sigma^{2}$ by the proposed method and
`Benchmark1' is getting smaller, while it is unaffected under `Benchmark2'.
These phenomenons suggest that our proposed method and `Benchmark1' can reap the benefit from the device diversity, while `Benchmark2'
cannot. By comparing the proposed method and `Benchmark1', we can
observe that our proposed method is inferior to `Benchmark1', since
the optimized beamforming vector enjoy the benefit of multiple antennas
at each iteration is also a vector in the complex unit sphere. Besides,
the gap is shrinking as the increase of $K$, which means that the
proposed method can approach `Benchmark1'. It is worth noting that
the computation for two benchmarks to obtaining $\mathbf{m}$
is intensive, which is negligible in our proposed random beamforming
design.


\begin{figure}
\subfigure[]{\includegraphics[width=0.45\columnwidth]{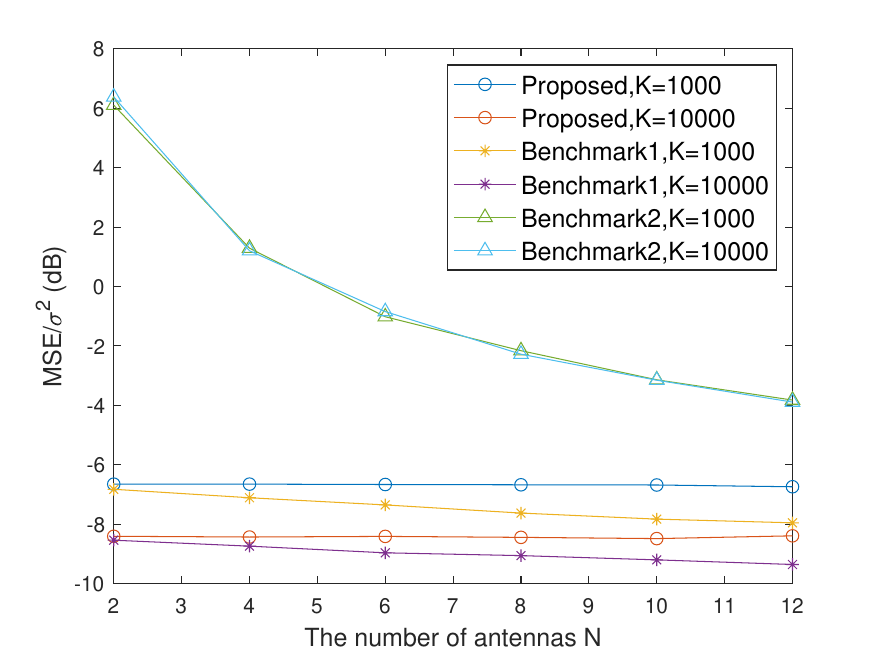}}
\subfigure[]{\includegraphics[width=0.45\columnwidth]{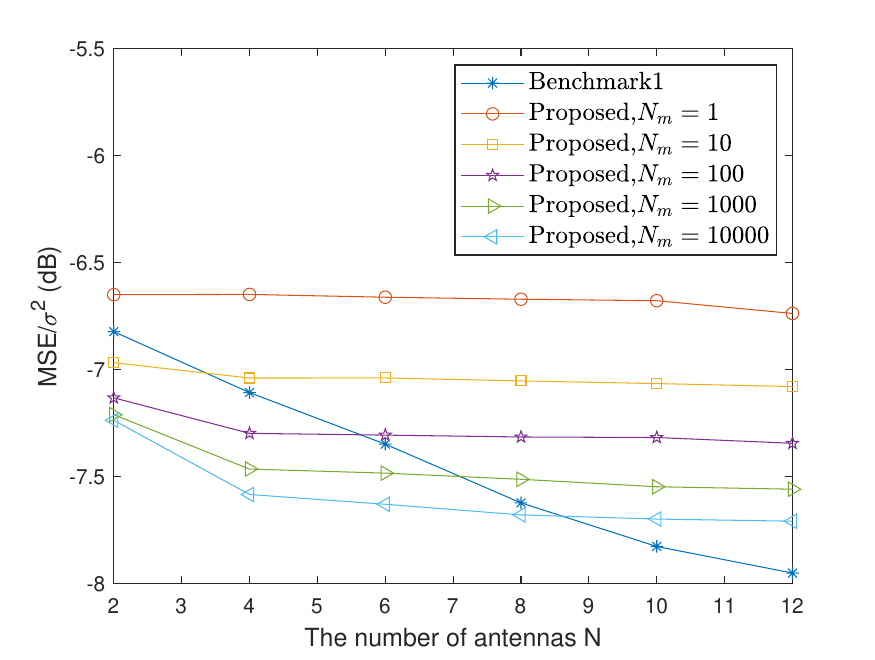}}
\captionsetup{font={small}}\caption{MSE performance versus number of antennas. (a) Comparison of the proposed scheme with two benchmarks with different number of devices.  (b) Proposed scheme with a fixed number of devices $K=1000$ and different randomlization numbers.}
\label{fig:MSE_Ant}
\end{figure}

To verify the MSE performance improvement by the refined method, we
compare the refined methods with `Benchmark1' in the scenario consisting
of $K=10^3$ devices and an aggregator with different number of antennas. Fig.
\ref{fig:MSE_Ant}(b) gives the MSE performance of `Benchmark1' and
the refined method under $N_{m}=1,10,10^{2},10^{3},10^{4}$ randomizations.
The figure shows that the MSE performance can be improved by sampling
the vectors from the unit sphere for multiple times, which verifies
the effectiveness of the proposed refined method. Besides, the performance
improving rates become slower w.r.t. the number of randomizations $N_{m}$.
When there are $N=2,4,6$ antennas at the aggregator, the performance
under $N_{m}=10$, $10^2$ and $10^3$ randomizations can respectively
achieve better performance than `Benchmark1'. A larger $N_{m}$ is
required in order to get the performance outperforming `Benchmark1'
if more antennas are deployed. This can be explained such that the
space of the unit sphere is much larger as the increase of $N$.

\subsection{Maximum Number of Selected Devices Performance}

This subsection presents the performance of the number of the selected
devices w.r.t. the threshold $\overline{\mathrm{MSE}}/\sigma^{2}$,
the range of which is set $[-6,+6]$ dB. We consider three settings
where there are $N=4$ antennas at the aggregator and $K=50,100$, and $150$
devices in the system. We compare the proposed random aggregate beamforming
based and the refined methods with the the novel DC and $\ell_{1}$+SDR
algorithms, which are label as `DC' and `$\ell_{1}$+SDR' respectively.
As shown in Fig. \ref{fig:maxnum}, more devices are selected if generating
more random aggregate beamforming vectors, which is referred to as
the refined method. However, the improvement rate becomes slower as
the increase of $N_{m}$. The $\ell_{1}$+SDR method is shown to have
the poorest performance except for the case of $\overline{\mathrm{MSE}}/\sigma^{2}=-6$
dB. When the threshold is set to $-6$ dB, the $\ell_{1}$+SDR method
achieves slightly better performance than our proposed random aggregate
beamforming based method, but shows inferior performance compared
with the proposed refined one. The number of selected devices $\left|\mathcal{S}\right|$
shows a growing trend w.r.t. the threshold $\overline{\mathrm{MSE}}/\sigma^{2}$.
In a small $\overline{\mathrm{MSE}}/\sigma^{2}$ region, `DC' method
achieves the best performance compared with the other two methods.
In the intermediate values of $[-2,+2]$, the proposed method outperforms `DC' method,
the gap of which is shrunk in the large region. Additionally, Fig.
\ref{fig:maxnum} shows that the advantage of our proposed method
becomes more evident when there are more devices. It is worth noting
here that the computational complexity of our proposed method is $\mathcal{O}\left(N_{m}K\right)$,
which is negligible compared with both `DC' and `$\ell_{1}$+SDR'
reference methods.
\begin{figure}
	\centering
	\subfigure[] {\includegraphics[width=0.31\columnwidth]{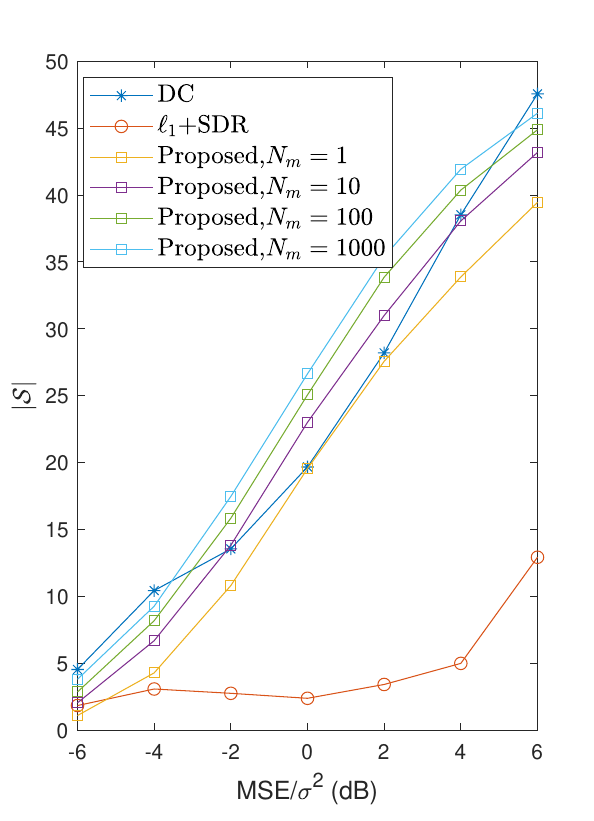}}
	\subfigure[] {\includegraphics[width=0.31\columnwidth]{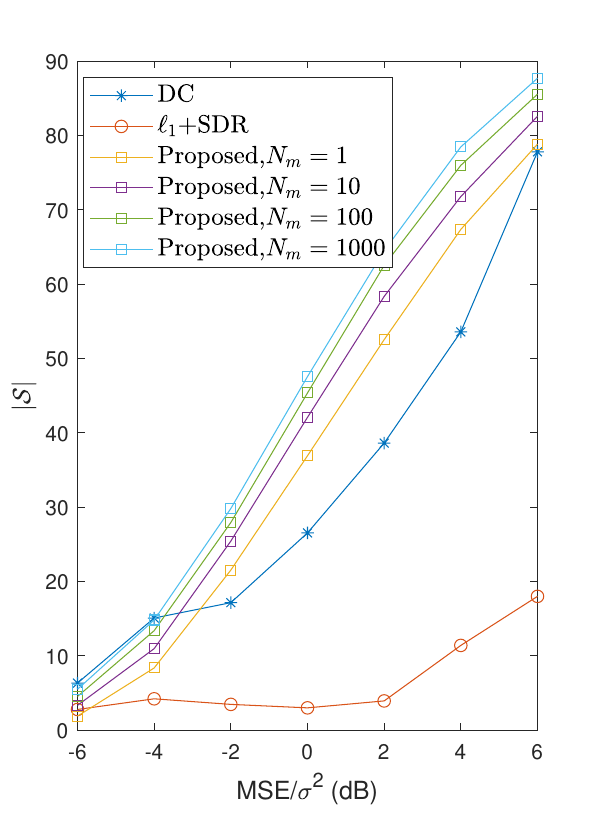}}
	\subfigure[] {\includegraphics[width=0.31\columnwidth]{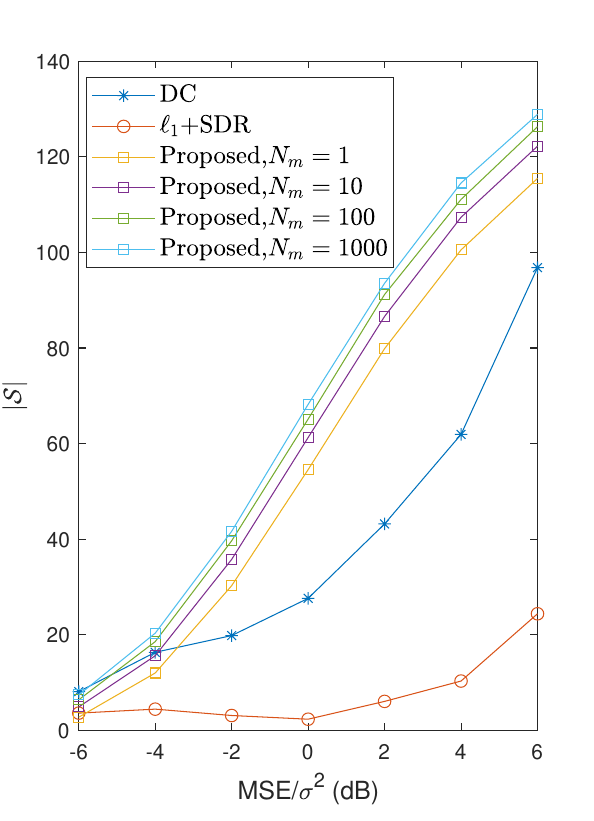}}
	\vspace{-0.2em}
	\captionsetup{font={small}}\caption{Number of selected devices $\left|\mathcal{S}\right|$
		versus MSE with different number of devices. (a) $K=50$. (b)
		$K=100$. (c) $K=150$.}
\label{fig:maxnum}
\end{figure}

\subsection{Learning Performance}

The previous subsection presents the performance from the communication
perspective. This subsection illustrates the impact of the MSE and
the selected devices on the learning performance. We utilize FL framework
to train the standard image classification tasks on two well-known
datasets, i.e., MNIST10 and CIFAR10. MNIST10 dataset consisting of
$10$ classes of black-and-white handwritten digital picture is easier
to learn, where a multilayer perception (MLP) neural network is adopted.
By contrast, the color pictures of CIFAR10 are much difficult to learn,
and thereby a more complex ResNet18 neural network is used. For simplicity, the effect of the aggregate error on the learning performance is modelled as
the model retransmission, the probability of which is $p=1-\exp(-a\mathrm{MSE}/\sigma^{2})$
where parameter $a$ is set to $1$.

For the impact of MSE on the learning performance, we consider the
scenario where $S=10$ devices are selected from $K=100$ devices
for model aggregation at the aggregator equipped with $N=4$ antennas.
The obtained average MSE performance $\mathrm{MSE}/\sigma^{2}$ under
the proposed method, `Benchmark1' and `Benchmark2' methods are respectively
$0.3221$, $0.3410$ and $1.3531$. Fig. \ref{fig:mnist_noniid_mse}
presents the testing accuracy on both datasets, which are non-i.i.d.
distributed among the devices, which shows the slowest convergence
rate under `Benchmark2'. This can be explained such that a smaller
$\mathrm{MSE}/\sigma^{2}$ leads to a smaller retransmission probability
and more validate training rounds. Since the MSE performance under
our proposed method and `Benchmark1' have approximately the same value,
they have almost the same validate training epochs and thereby approximately
the same test accuracy.

\begin{figure}
	\centering
	\subfigure[] {\includegraphics[width=0.48\columnwidth]{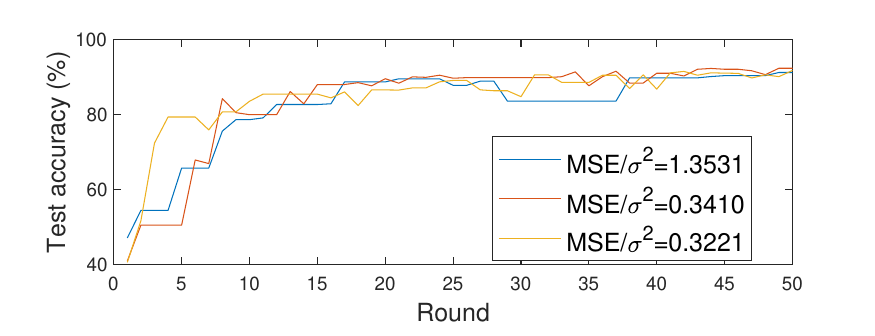}}
	\subfigure[] {\includegraphics[width=0.48\columnwidth]{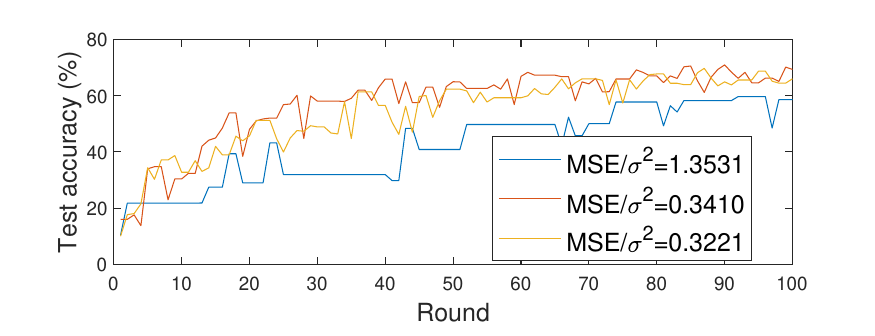}}
	\captionsetup{font={small}}\caption{Impact of MSE performance on the
		test accuracy averaged. (a) Non-i.i.d. MNIST. (b) Non-i.i.d. CIFAR10.}
\label{fig:mnist_noniid_mse}
\end{figure}

We also conduct the simulations on FL performance in terms of the
number of selected devices under different methods. There are $K=100$
devices, $N=4$ antennas at the aggregator. The MSE performance threshold
is set to $\overline{\mathrm{MSE}}/\sigma^{2}=-2$ dB. In this case,
the number of selected devices under `DC', `$\ell_{1}$+SDR' and the
refined methods with $N_{m}=10^3$ are $17$, $3$ and $30$ respectively.
Fig. \ref{fig:mnist_noniid_num} displays the testing accuracy on
non-i.i.d. MNIST and CIFAR10 datasets. As can be seen, the classification
accuracy converges faster if more devices are selected for local model
update and global model aggregation. This can be explained such that the average gradient computed by a small number of devices are may far from the true one, which increases the performance fluctuation and slows the convergence rates.

\begin{figure}
\centering
\subfigure[] {\includegraphics[width=0.48\columnwidth]{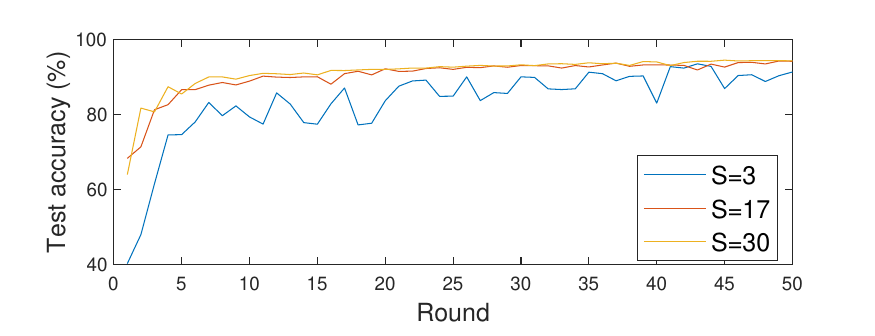}}
\subfigure[] {\includegraphics[width=0.48\columnwidth]{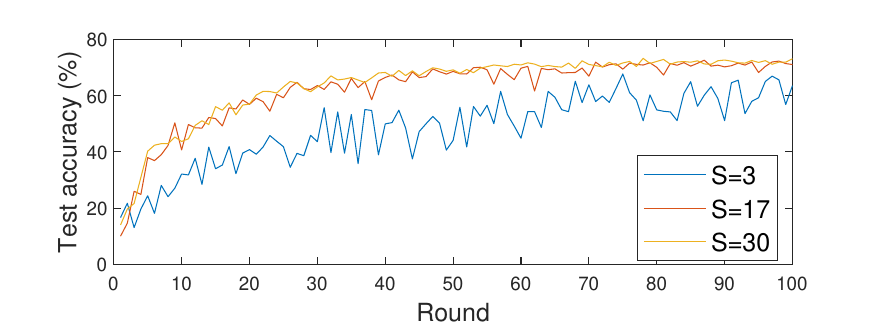}}
\captionsetup{font={small}}\caption{Impact of $\left|\mathcal{S}\right|$
on the test accuracy. (a) Non-i.i.d. MNIST. (b) Non-i.i.d. CIFAR10.}
\label{fig:mnist_noniid_num}
\end{figure}

\section{Conclusion}

In this paper, we investigate edge intelligence in a large-scale
wireless network, where the FL framework is adopted to train a shared
model. To aggregate the local models, the AirComp technique is adopted,
which aggregates the local models in an analog manner. The design of device selection
and model transmission schemes is critical for both learning and communication performances. We studied joint device selection and aggregate beamforming design
with the two objectives of the aggregate error minimization and the number
of selected devices maximization. To ease the computational complexity
in a large-scale system, we proposed a random aggregate beamforming-based
scheme, the core idea is to sample a vector from a complex
unit sphere first and select the devices afterwards. When the number
of the devices goes to infinity, we theoretically proved that the performance
obtained by the proposed methods approached the optimal MSE performance
for the aggregate error minimization problem. For the number of the
selected devices maximization, the analysis also gave the interval
and the average value of the number of the selected devices. When
the number of devices is much smaller than infinity, the refined method
was proposed aiming at performance improvement, the effectiveness
of which was analyzed. Simulation results demonstrated the effectiveness
of the proposed methods, and confirmed the theoretical analysis.

\bibliographystyle{IEEEtran}
\bibliography{Random_TCOM}

\end{document}